\documentclass[english,a4paper,11pt]{article}
\usepackage{fullpage}

\usepackage{amsmath, amsxtra, amsfonts, amssymb, amstext}
\usepackage{amsthm}
\usepackage{booktabs}
\usepackage{graphics,color}
\usepackage{graphicx}
\usepackage[usenames,dvipsnames,svgnames,table]{xcolor}
\usepackage{babel}
\usepackage[noadjust]{cite}
\usepackage{url}
\usepackage{paralist}

\usepackage[colorlinks]{hyperref}
\definecolor{linkblue}{rgb}{0.1,0.1,0.8}
\hypersetup{colorlinks=true,linkcolor=linkblue,filecolor=linkblue,urlcolor=linkblue,citecolor=linkblue}

\usepackage[algo2e,ruled,vlined,linesnumbered]{algorithm2e}

\SetKwFor{Function}{function}{is}{end function}
\newcommand{\assign}{\leftarrow}
\SetKw{Input}{input}
\SetKw{Output}{output}
\SetKw{Initialization}{initialization}
\SetCommentSty{default}

\newcommand{\ignore}[1]{}
\newcommand{\Advance}{\mathit{Advance}}

\newtheorem{theorem}{Theorem}[section]
\newtheorem{lemma}[theorem]{Lemma}

\newtheorem{corollary}[theorem]{Corollary}

\newcommand{\N}{\mathbb{N}}
\newcommand{\R}{\mathbb{R}}

\renewcommand{\epsilon}{\varepsilon}

\newcommand{\F}{\mathcal{F}}

\newcommand{\Ha}{\mathcal{H}}
\renewcommand{\lg}{\log}

\DeclareMathOperator{\E}{\mathbb{E}}
\DeclareMathOperator{\id}{id}

\newcommand{\leadingones}{\textsc{LeadingOnes}\xspace}
\newcommand{\ourproblem}{\textsc{HiddenPermutation}\xspace}

\definecolor{orange}{rgb}{1,0.5,0}

\newcommand{\RandBinSearch}{\texttt{RandBinSearch}}\newcommand{\subroutineone}{\texttt{RandBinSearch}}

\newcommand{\SizeReduction}{\texttt{SizeReduction}}
\newcommand{\ReductionStep}{\texttt{ReductionStep}}

\newcommand{\set}[2]{\{ #1 \, \mid \, #2 \}}
\newcommand{\sset}[1]{\{ #1 \}}
\newcommand{\abs}[1]{| #1 |}

\newcommand{\is}{{\textbf{i}}}

\newcommand{\eps}{\varepsilon}

\renewcommand{\mod}{\textrm{ mod }}
\renewcommand{\phi}{\varphi}
\newcommand{\tsLB}{\textbf{t}}
\DeclareMathOperator{\con}{Con}
\newcommand{\vs}{\ell}

\begin{document}

\title{The Query Complexity of a Permutation-Based Variant of Mastermind}

\author{
Peyman Afshani$^1$,
Manindra Agrawal$^2$,
Benjamin Doerr$^3$,\\
Carola Doerr$^{4}$\thanks{Corresponding author. Carola.Doerr@mpi-inf.mpg.de}, 
Kasper Green Larsen$^1$,
Kurt Mehlhorn$^5$
}
\date{
	$^1$MADALGO\footnote{Center for Massive Data Algorithms, a Center of the Danish National Research Foundation, Denmark}, Department of Computer Science, Aarhus University, Denmark\\
$^2$Indian Institute of Technology Kanpur, India\\
$^3$\'Ecole Polytechnique, CNRS, Laboratoire d'Informatique (LIX), Palaiseau, France\\
$^4$Sorbonne Universit\'e, CNRS, LIP6, Paris, France\\
$^5$Max Planck Institute for Informatics, Saarbr\"ucken, Germany}
\maketitle

\begin{abstract}
We study the query complexity of a permutation-based variant of the guessing game Mastermind. In this variant, the secret is a pair $(z,\pi)$ which consists of a binary string $z \in \{0,1\}^n$ and a permutation $\pi$ of $[n]$. The secret must be unveiled by asking queries of the form $x \in \{0,1\}^n$. For each such query, we are returned the score 
\[ f_{z,\pi}(x):= \max \{ i \in [0..n]\mid \forall j \leq i: z_{\pi(j)} = x_{\pi(j)}\}\,;\]
i.e., the score of $x$ is the length of the longest common prefix of $x$ and $z$ with respect to the order imposed by $\pi$. The goal is to minimize the number of queries needed to identify $(z,\pi)$. 
This problem originates from the study of black-box optimization heuristics, where it is known as the \textsc{LeadingOnes} problem. 

In this work, we prove matching upper and lower bounds for the deterministic and randomized query
complexity of this game, which are $\Theta(n \log n)$ and $\Theta(n \log \log n)$, respectively. 
\end{abstract}




\section{Introduction}
\label{sec:introduction}

\emph{Query complexity}, also referred to as \emph{decision tree complexity}, is one
of the most basic models of computation. We aim at learning an
unknown object (a secret) by asking queries of a certain type. The cost of
the computation is the number of queries made until the secret is unveiled. All
other computation is free. 
Query complexity is one of the standard measures in computational complexity theory. 

A related performance measure can be found in the \emph{theory of black-box optimization}, where we are asked to optimize a function $f:S\rightarrow \R$ without having access to it other than by evaluating (``querying'') the function values $f(x)$ for solution candidates $x \in S$. In black-box optimization, the performance of an algorithm on a class $\F$ of functions is measured by the number of function evaluations that are needed until, for any unknown function $f \in \F$, an optimal search point is evaluated for the first time. 

\subsection{Problem Description and Summary of Main Results}\label{sec:intro:problemdescription}

In this work, we consider the query complexity of the following problem.

Let $S_n$ denote the set of permutations of $[n] := \{1, \ldots, n\}$; let $[0..n]:=\{0,1, \ldots, n\}$.
Our problem is that of learning a hidden permutation 
$\pi \in S_n$ together with a hidden bit-string $z \in \{0,1\}^n$
through queries of the following type. A query is again a bit-string $x \in
\{0,1\}^n$. As answer we receive the length of the longest common prefix of $x$
and $z$ in the order of $\pi$, which we denote by 
\begin{align*} 
f_{z,\pi}(x) := 
\max \{ i \in [0..n]\mid \forall j \leq i: z_{\pi(j)} = x_{\pi(j)}\}\,.
\end{align*}
We call this problem the $\ourproblem$ problem. It can be viewed as a guessing game like the well-known Mastermind 
problem (cf. Section~\ref{sec:intro:motivation}); however, the secret is now a pair $(z,\pi)$ and not just a
string. The problem is a standard benchmark problem in black-box optimization.

It is easy to see (see Section~\ref{sec:definitions}) that learning one part of the secret (either $z$ or $\pi$) is no easier (up to $O(n)$ questions) than
learning the full secret.
It is also not difficult to see that $O(n \log n)$ queries suffice deterministically to
unveil the secret (see Section~\ref{sec:deterministic}).   
Doerr and Winzen~\cite{DoerrW11EA} showed that randomization
allows to beat this bound. They gave a randomized algorithm with 
$O(n \log n / \log\log n)$ expected complexity. The information-theoretic lower
bound is only $\Theta(n)$ as the answer to each query is a number between zero
and $n$ and hence may reveal as many as $\log n$ bits. We show that 
\begin{compactenum}[\hspace{\parindent}(1)]
\item the deterministic
query complexity is $\Theta(n \log n)$,
cf.~Section~\ref{sec:deterministic}, and that
\item the randomized query complexity is 
$\Theta(n \log\log n)$, cf.~Sections~\ref{sec:upper} and~\ref{sec:lower}. 
\end{compactenum}
Both
upper bound strategies are efficient, i.e., they can be implemented in polynomial
time. The lower bound is established by a (standard)
adversary argument in the deterministic case and by a potential function
argument in the randomized case. We remark that for many related problems (e.g., sorting, Mastermind, and many coin weighing problems) the asymptotic query complexity, or the best known lower bound for it, equals the information-theoretic lower bound. For our problem, deterministic and randomized query complexity differ, and the randomized query complexity exceeds the information theoretic lower bound. 
The randomized upper and lower bound require
non-trivial arguments. In Section~\ref{sec:definitions} we derive auxiliary
results, of which some are interesting in their
own right. For example, it can be decided efficiently whether a sequence
of queries and answers is consistent. 

A summary of the results presented in this work previously appeared in~\cite{AfshaniADLMW12}.

\subsection{Origin of the Problem and Related Work}
\label{sec:intro:motivation} 

The $\ourproblem$ problem can be seen as a \emph{guessing game.} An archetypal and well-studied representative of guessing games is \emph{Mastermind}, a classic two-player board game from the Seventies. In the Mastermind game, one player chooses a secret code $z \in [k]^n$. The role of the second player is to identify this code using as few guesses as possible. For each guess $x \in [k]^n$ she receives the number of positions in which $z$ and $x$ agree (and in some variants also the number of additional ``colors'' which appear in both $x$ and $z$, but in different positions). 

The Mastermind game has been studied in different context. Among the most prominent works is a study of Erd\H os and R\'enyi~\cite{Erd63}, where the 2-color variant of Mastermind is studied in the context of a coin-weighing problem. Their main result showed that the query complexity of this game is $\Theta(n/\log n)$. This bound has subsequently been generalized in various ways, most notably to the Mastermind game with $k=n$ colors. Using similar techniques as Erd\H os and R\'enyi, Chv\'atal~\cite{Chvatal83} showed an upper bound of $O(n \log n)$ for the query complexity of this game. This bound has been improved $O(n \log\log n)$~\cite{DoerrDST16}. The best known lower bound is the information-theoretic linear one. This problem is open for more than 30 years.

Our permutation-based variant of Mastermind has its origins in the field of evolutionary computation. There, the $\leadingones$ function 
$\{0,1\}^n \rightarrow [0..n], x \mapsto \max \{ i \in [0..n] \mid \forall j \leq i: x_j=1\}\,$
counting the number of initial ones in a binary string of length $n$, 
is a commonly used benchmark function both for experimental and theoretical analyses
(e.g.,~\cite{Rudolph97}). It is studied as one of the simplest examples of a unimodal non-separable function.  

An often desired feature of general-purpose optimization algorithms like genetic and evolutionary
algorithms is that they should be oblivious of problem representation. In the context of optimizing functions of the form $f:\{0,1\}^n \rightarrow \R$, this is often formalized as an \emph{unbiasedness} restriction, which requires that the performance of an unbiased algorithm is identical for all composed functions $f \circ \sigma$ of $f$ with an automorphism $\sigma$ of the $n$-dimensional hypercube. Most standard evolutionary algorithms as well many other commonly used black-box heuristics like local search variants, simulated annealing, etc. are unbiased. Most of them query an average number of $\Theta(n^2)$ solution candidates until they find the maximum of $\leadingones$, see, e.g.,~\cite{DrosteJW02}. 

It is not difficult to see that 
$$\left\{\leadingones \circ \sigma \mid \sigma \text{ is an automorphism of } \{0,1\}^n \right\}
= 
\left\{f_{z,\pi} \mid z \in \{0,1\}^n, \pi \in S_n \right\}, 
$$ 
showing that $\ourproblem$ generalizes the $\leadingones$ function by indexing the bits not
from left to right, but by an arbitrary permutation, and by swapping the interpretation of $0$ and $1$ for indices $i$ with $z_i=0$. 

The first to study the query complexity of $\ourproblem$ were Droste, Jansen, and Wegener~\cite{DrosteJW06}, who introduced the notion of \emph{black-box complexity} as a measure for the difficulty of black-box optimization problems. As mentioned, the main objective in black-box optimization is the identification of an optimal solution using as few function evaluations as possible. Denoting by $T_{A}(f)$ the number of queries needed by algorithm $A$ until it queries for the first time a search point $x \in \arg\max f$ (this is the so-called \emph{runtime} or \emph{optimization time} of $A$), the \emph{black-box complexity} of a collection $\F$ of functions is 
$$\inf_{A} \max_{f \in \F} \E[T_A(f)],$$ 
the best (among all algorithms) worst-case (with respect to all $f \in \F$) expected runtime. 
Black-box complexity is
essentially query-complexity, with a focus on optimization rather than on learning. A survey on this topic can be found in~\cite{Doerr18BBC}.

Droste et al.~\cite{DrosteJW06} considered only the $0$/$1$-invariant class 
$\left\{f_{z,\id} \mid z \in \{0,1\}^n \right\}$ 
(where $\id$ denotes the identity permutation). They showed that the black-box complexity of this function class is $n/2 \pm o(n)$. 

The first bound for the black-box complexity of the full class $\ourproblem$ was presented in~\cite{LehreW12} by Lehre and Witt,  who showed that any so-called \emph{unary unbiased black-box algorithm} needs $\Omega(n^2)$ steps, on average, to solve the $\ourproblem$ problem. In~\cite{DoerrJKLWW11} it was then shown that already with binary distributions one can imitate a binary search algorithm (similar to the one presented in the proof of Theorem~\ref{thm:deterministic}), thus yielding an $O(n \log n)$ black-box algorithm for the $\ourproblem$ problem. The algorithm achieving the $O(n \log n / \log\log n)$ bound mentioned in Section~\ref{sec:intro:problemdescription} can be implemented as a ternary unbiased one~\cite{DoerrW11EA}. This bound, as mentioned above, was the previously best known upper bound for the black-box complexity of $\ourproblem$. 

In terms of lower bounds, the best known bound was the linear information-theoretic one. However, more recently a stronger lower bound has been proven to hold for a restricted class of algorithms. More precisely, Doerr and Lengler~\cite{DoerrL17LO} recently proved that the so-called $(1+1)$ \emph{elitist black-box complexity} is $\Omega(n^2)$, showing that any algorithm beating this bound has to keep in its memory more than one previously evaluated search point or has to make use of information hidden in non-optimal search points. It is conjectured in~\cite{DoerrL17LO} that the (1+1) memory-restriction alone (which allows an algorithm to store only one previously queried search point in its memory and to evaluate in each iteration only one new solution candidate) already causes the $\Omega(n^2)$ bound, but this conjecture stands open. We note that many local search variants, including simulated annealing, are $(1+1)$ memory-restricted.

\section{Preliminaries}
\label{sec:definitions}
\label{SEC:DEFINITIONS}

For all positive integers $k \in \N$ we define $[k]:=\{1,\ldots,k\}$ and $[0..k]:=[k] \cup \{0\}$. 
By $e^n_k$ we denote the $k$th unit vector $(0,\ldots,0,1,0,\ldots,0)$ of length $n$. 
For a set $I \subseteq [n]$ we define $e^n_I:=\sum_{i \in I}{e^n_i}=\oplus_{i \in I}{e^n_i}$, where $\oplus$ denotes the bitwise exclusive-or.
We say for two bitstrings $x,y \in \{0,1\}^n$ that we \emph{created $y$ from $x$ by flipping $I$} or that we \emph{created $y$ from $x$ by flipping the entries in position(s) $I$} if $y=x \oplus e^n_I$.
By $S_n$ we denote the set of all permutations of $[n]$.
For $r\in\R_{\geq 0}$, let 
$\lceil r \rceil := \min\{n\in\N_0 \mid n\geq r\}$. 
and $\lfloor r \rfloor := \max\{n\in\N_0 \mid n \leq r\}$. 
To increase readability, we sometimes omit the $\lceil \cdot \rceil$ signs; that is, whenever we write $r$ where an integer is required, we implicitly mean $\lceil r \rceil$.

Let $n \in \N$. 
For $z \in \{0,1\}^n$ and $\pi \in S_n$ we define 
\begin{align*} 
f_{z,\pi}: 
\{0,1\}^n \rightarrow [0..n], 
x \mapsto
\max \{ i \in [0..n]\mid \forall j \leq i: z_{\pi(j)} = x_{\pi(j)}\}\,.
\end{align*}
$z$ and $\pi$ are called the \emph{target string} and \emph{target permutation}
of $f_{z,\pi}$, respectively. We want to identify target string and
permutation by asking queries $x^i$, $i = 1$, $2$, \ldots, and evaluating the answers (``scores'') $s^i = f_{z,\pi}(x^i)$. 
We may stop after
$t$ queries if there is only a \emph{single} pair $(z,\pi) \in \{0,1\}^n \times
S_n$ with $s^i = f_{z,\pi}(x^i)$ for $1 \le i \le t$. 
 
A \emph{deterministic strategy} for the $\ourproblem$ problem is a tree of outdegree
$n+1$ in which a query in $\sset{0,1}^n$ is associated with every node of
the tree. The search starts as the root. If the search reaches a node, the
query associated with the node is asked, and the
search proceeds to the child selected by the score. 
The complexity of a strategy on input $(z,\pi)$ is the
number of queries required to identify the secret, and the complexity of a deterministic
strategy is the worst-case complexity of any input. That is, the complexity is the height of the search tree.

A \emph{randomized strategy} is a
probability distribution over deterministic strategies. The complexity of a
randomized strategy on input $(z,\pi)$ is the expected
number of queries required to identify the secret, and the complexity of a randomized
strategy is the worst-case complexity of any input. The probability
distribution used for our randomized upper bound is a product distribution in
the following sense: a probability distribution over $\sset{0,1}^n$ is associated with every node of
the tree. The search starts as the root. In any node, the query is selected
according to the probability distribution associated with the node, and the
search proceeds to the child selected by the score.

We remark that knowing $z$ allows us to determine $\pi$ with $n-1$ queries $z
\oplus e^n_i$, $1 \le i < n$. Observe that $\pi^{-1}(i)$ equals $f_{z,\pi}(z
\oplus e^n_i)+1$. Conversely, knowing the target permutation $\pi$ we
can identify $z$ in a linear number of guesses. If our query $x$ has a score of
$k$, all we need to do next is to query the string $x'$ that is created from
$x$ by flipping the entry in position $\pi(k+1)$.  Thus, learning one part of the secret is no easier (up to $O(n)$ questions) than
learning the full.  

A simple information-theoretic argument gives an $\Omega(n)$ lower bound for the deterministic query complexity and, together with Yao's minimax principle~\cite{Yao77}, also for the randomized complexity. The
\emph{search space} has size $2^n n!$, since the unknown secret is an element of
$\{0,1\}^n \times S_n$. A deterministic strategy is a tree with outdegree $n+1$
and $2^n n!$ leaves. The maximal and average depth of any such tree is
$\Omega(n)$. 

Let
$\Ha:=(x^i, s^i)_{i=1}^t$ be a sequence of queries
$x^i \in \{0,1\}^n$ and scores $s^i \in [0..n]$. 
We call $\Ha$ a \emph{guessing history.} A secret $(z,\pi)$ is \emph{consistent with
$\Ha$} if $f_{z,\pi}(x^i)=s^i$ for all $i \in [t]$. $\Ha$ is \emph{feasible} if there exists a secret consistent with it.

An observation crucial in our proofs is the fact that a vector $(V_1,
\ldots, V_n)$ of subsets of $[n]$, together with a top score query $(x^*,s^*)$,
captures the total knowledge provided by a \emph{guessing history} 
$\Ha=(x^i, s^i)_{i=1}^t$ about the set of secrets consistent with $\Ha$. 
We will call $V_j$ the \emph{candidate set} for position $j$; 
$V_j$ will contain all indices $i \in [n]$ for which the following simple rules (1) to (3)
do not rule out that $\pi(j)$ equals $i$. Put differently, $V_j$ contains the set of values that are still possible images for $\pi(j)$, given the previous queries.

\begin{theorem}
\label{thm:knowledge}
Let $t \in \N$, and let $\Ha=(x^i, s^i)_{i=1}^t$ be a guessing history. 
Construct the candidate sets $V_1, \ldots, V_n \subseteq [n]$ according to the following
rules: 
\begin{compactenum}[\hspace{\parindent}(1)]
\item If there are $h$ and $\ell$ with $j \le s^h \le s^{\ell}$ and $x_i^h \not= x_i^{\ell}$,
then $i \not\in V_j$.
\item If there are $h$ and ${\ell}$ with $s=s^h = s^{\ell}$ and $x_i^h \not= x_i^{\ell}$,
then $i \not\in V_{s+1}$. 
\item If there are $h$ and ${\ell}$ with $s^h < s^{\ell}$ and $x_i^h = x_i^{\ell}$,
then $i \not\in V_{s^h+1}$. 
\item If $i$ is not excluded by one of the rules above, then $i \in
  V_j$. 
\end{compactenum}
Furthermore, let $s^*:=\max\{ s^1, \ldots, s^t\}$
and let 
$x^*=x^j$ for some $j$ with $s^j=s^*$.

Then a pair $(z,\pi)$ is consistent with $\Ha$ if and only if
\begin{compactenum}[\hspace{\parindent}(a)]
\item $f_{z,\pi}(x^*)=s^*$ and 
\item $\pi(j) \in V_j$ for all $j \in [n]$.
\end{compactenum}
\end{theorem}

\begin{proof} 
Let $(z, \pi)$ satisfy conditions (a) and (b).
We show that $(z, \pi)$ is consistent with~$\Ha$.
To this end, let $h \in [t]$, let $x = x^h$, $s = s^h$, and $f :=f_{z,\pi}(x)$.
We need to show $f = s$.

Assume $f < s$. 
Then $z_{\pi(f+1)} \neq x_{\pi(f+1)}$.
Since $f+1 \leq s^*$, this together with~(a) implies 
$x_{\pi(f+1)} \neq x^*_{\pi(f+1)}$.
Rule (1) yields 
$\pi(f+1) \notin V_{f+1}$; a contradiction to (b).

Similarly, if we assume
$f > s$, then 
$x_{\pi(s+1)} = z_{\pi(s+1)}$. 
We distinguish two cases. 
If $s < s^*$, then by 
condition (a) we have
$x_{\pi(s+1)} = x^*_{\pi(s+1)}$. 
By rule (3) this implies
$\pi(s+1) \notin V_{s+1}$; a contradiction to (b).

On the other hand, if $s = s^*$, then
$x_{\pi(s+1)} 
= z_{\pi(s+1)}
\neq x^*_{\pi(s+1)}$ by (a).
Rule (2) implies $\pi(s+1) \notin V_{\pi(s+1)}$, again contradicting (b).

Necessity is trivial.  
\end{proof}

We may also construct the sets $V_j$ incrementally. The following update rules
are direct consequences of Theorem~\ref{thm:knowledge}.
In the beginning, let $V_j:=[n]$, $1 \leq j \leq n$. 
After the first query, record the first query as $x^*$ and
  its score as $s^*$.
For all subsequent queries, do the following:
Let $I$ be the set of indices in which the current query $x$ and the current best query $x^*$ agree. 
Let $s$ be the score of $x$ and let $s^*$ be the score of
$x^*$. \smallskip

\begin{compactenum}[\hspace{\parindent}Rule 1:]
\item If $s < s^*$, then $V_i \assign V_i \cap I$ for $1 \le i \le s$ and $V_{s+1} \assign
V_{s+1} \setminus I$. 
\item If $s = s^*$, then $V_i \assign V_i \cap I$ for $1 \le i \le s^* +
1$.
\item If $s > s^*$, then $V_i \assign V_i \cap I$ for $1 \le i \le s^*$ and
$V_{s^*+1} \assign V_{s^* + 1} \setminus I$. We further replace $s^* \assign s$
and $x^* \assign x$. 
\end{compactenum}\smallskip

It is immediate from the update rules that the sets $V_j$ form a \emph{laminar family}; i.e., for $i < j$ either
$V_i \cap V_j = \emptyset$ or $V_i \subseteq V_j$. 
As a consequence of Theorem~\ref{thm:knowledge} we obtain a polynomial
  time test for the feasibility of histories. It gives additional insight in the meaning of the candidate sets $V_1, \ldots, V_n$.

\begin{theorem} 
\label{thm:consistency}
It is decidable in polynomial time whether a guessing history 
is feasible. 
Furthermore, we can efficiently compute the number of pairs 
$(z, \pi)~\in~\{0,1\}^n~\times~S_n$
consistent with it.
\end{theorem} 
\begin{proof}
We first show that feasibility can be checked in polynomial time.
Let $\Ha=(x^i,s^i)_{i=1}^t$ be given. 
Construct the sets $V_1, \ldots, V_n$ as described in Theorem~\ref{thm:knowledge}. 
Now construct a bipartite graph $G(V_1,\ldots,V_n)$ with
node set $[n]$ on both sides. Connect $j$ to all nodes in $V_j$ on the other
side. Permutations $\pi$ with $\pi(j) \in V_j$ for all~$j$ are in one-to-one
correspondence to 
perfect matchings in this graph. We recall that perfect matchings can be computed efficiently in bipartite graphs, e.g., using the Ford-Fulkerson algorithm, which essentially treats the matching problem as a network flow problem. If there is no perfect matching, the history
in infeasible. Otherwise, let $\pi$ be any permutation with $\pi(j) \in V_j$ for all $j$. We next construct
$z$. We use the obvious rules:\smallskip

\begin{compactenum}[\hspace{\parindent}(a)]
\item If $i = \pi(j)$ and $j \le s^h$ for some $h \in [t]$ then set $z_i
:= x_i^h$. 
\item  If $i = \pi(j)$ and $j = s^h + 1$ for some $h \in [t]$ then set $z_i := 1 - x_i^h$. 
\item If $z_i$ is not defined by one of the rules above, set it to an arbitrary
value.
\end{compactenum}\smallskip

\noindent We need to show that these rules do not lead to a contradiction. Assume
otherwise. There are three ways, in which we could get into a contradiction.
There is some $i \in [n]$ and some $x^h, x^\ell \in \{0,1\}^n$ \smallskip

\begin{compactenum}[\hspace{\parindent}(1)]
\item setting $z_i$ to opposite values by rule (a)
\item setting $z_i$ to opposite values by rule (b)
\item setting $z_i$ to opposite values by rules (b)
applied to $x^h$ and rule (a) applied to $x^{\ell}$. 
\end{compactenum}\smallskip

\noindent In each case, we readily derive a contradiction.
In the first case, we have 
$j \le s^h$, $j \le s^{\ell}$ and $x_i^h \not= x_i^{\ell}$. 
Thus $\pi(j) = i \not\in V_j$ by rule (1). 
In the
second case, we have $j = s^h+1 = s^{\ell}+1$ and $x_i^h \not= x^{\ell}_i$. Thus $i \not\in V_j$ by (2). 
In the third
case, we have $j = s^h + 1$, $j \le s^{\ell}$, and $x_i^h = x_i^{\ell}$. Thus $i \not\in V_j$ by (3). 

Finally, the pair $(z,\pi)$ defined in this way is clearly consistent with the
history. \medskip

Next we show how to efficiently compute the number of consistent pairs. We
recall Hall's condition for the existence of a perfect matching in a bipartite
graph. A perfect matching exists if and only if $\abs{ \bigcup_{j \in J} V_j} \ge
\abs{J}$ for every $J \subseteq [n]$. 
According to Theorem~\ref{thm:knowledge}, the guessing history $\Ha$ can be equivalently described by a state
$(V_1,\ldots, V_{s^* + 1}, x^*,s^*)$. 
How many pairs $(z,\pi)$ are compatible with this state?

Once we have chosen $\pi$, there are exactly $2^{n - (s^* + 1)}$ different
choices for $z$ if $s^* < n$ and exactly one choice if $s^* = n$. The
permutations can be chosen in a greedy fashion. We fix $\pi(1), \ldots,
\pi(n)$ in this order. The number of choices for $\pi(i)$ equals 
$\abs{V_i}$ minus the number of $\pi(j)$, $j < i$, lying in $V_i$. If $V_j$ is
disjoint from $V_i$, $\pi(j)$ never lies in $V_i$ and if $V_j$ is contained in
$V_i$, $\pi(j)$ is always contained in $V_i$. Thus the number of permutations
is equal to
\[                    
\prod_{1 \le i \le n} \left(\abs{V_i} - \abs{\set{j < i}{V_j
\subseteq V_i}}\right)\,.\]
It is easy to see that the greedy strategy does not violate Hall's condition.
\end{proof}

We note that it is important in the proof of Theorem~\ref{thm:consistency} that the sets $V_j$ form a laminar family; 
counting the number of perfect matchings in a general bipartite graph is
\#P-complete~\cite{Valiant79}. 

From the proof can also derive which values in $V_i$ are actually possible as value for $\pi(i)$. This will be described in the next paragraph, which a reader only interested in the main results can skip without loss. 
A value $\ell \in V_i$ is feasible if there is a perfect matching in the graph
$G(V_1,\ldots,V_n)$ containing the edge $(i,\ell)$. The existence of such a
matching can be decided in polynomial time; we only need to test for a perfect
matching in the graph $G \setminus \sset{i,\ell}$. Hall's condition says that
there is no such perfect matching if there is a set $J \subseteq [n] \setminus
\{i\}$ such that $\abs{\bigcup_{j \in J} V_j \setminus \{\ell\}} < \abs{J}$. Since $G$
contains a perfect matching (assuming a consistent history), this implies 
 $\abs{\bigcup_{j \in J} V_j} = \abs{J}$; i.e.,
  $J$ is tight for Hall's condition. We have now shown: Let $\ell \in V_i$. Then
  $\ell$ is infeasible for $\pi(i)$ if and only if there is a tight set $J$ with $i
  \not\in J$ and $\ell \in \bigcup_{j \in J} V_j$. Since the $V_i$ form a laminar
  family, minimal tight sets have a special form; they consist of an $i$ and all
  $j$ such that $V_j$ is contained in $V_i$. In the counting formula for the
  number of permutations such $i$ are characterized by $\abs{V_i} - \abs{\set{j < i}{V_j
\subseteq V_i}}= 1$. In this situation, the elements of $V_i$ are infeasible
for all $\pi(j)$ with $j > i$. We may subtract $V_i$ from each $V_j$ with $j >
i$. 

If Hall's condition is tight for some $J$, i.e., $\abs{\bigcup_{j \in J} V_j} =
\abs{J}$, we can easily learn how $\pi$ operates on $J$. 
We have $V_j = [n]$ for
$j > s^*+1$ and hence the largest index in $J$ is at most $s^* + 1$. 
Perturb $x^*$ by flipping each bit in $\bigcup_{j \in
J} V_j$ exactly once. The scores determine the permutation.

\section{Deterministic Complexity}
\label{sec:deterministic}

We settle the deterministic query complexity of the $\ourproblem$
problem. The upper and lower bound match up to a small constant factor. Specifically, we prove

\begin{theorem}
\label{thm:deterministic}
The deterministic query complexity of the $\ourproblem$ problem  with $n$ positions is $\Theta(n \log n)$.
\end{theorem}

\begin{proof}
The \emph{upper bound} is achieved by an algorithm that resembles binary search
and iteratively identifies $\pi(1), \ldots, \pi(n)$ and the corresponding bit
values $z_{\pi(1)}, \ldots, z_{\pi(n)}$: We start by setting the set $V$ of
candidates for $\pi(1)$ to $[n]$ and by determinining a string $x$ with score
0; either the all-zero-string or the all-one-string will  work. We iteratively
reduce the size of $V$ keeping the invariant that $\pi(1) \in V$. We select an
arbitrary subset $F$ of $V$ of size $\abs{V}/2$ and create $y$ from $x$ by
flipping the bits in $F$. If $f_{z,\pi}(y) = 0$, $\pi(1) \not\in F$, if If
$f_{z,\pi}(y) > 0$, $\pi(1) \in F$. In either case, we essentially halve the
size of the candidate set. Continuing in this way, we determine $\pi(1)$ in 
$O(\log n)$ queries. Once $\pi(1)$ and
$z_{\pi(1)}$ are known, we iterate this strategy on the remaining bit positions
to determine $\pi(2)$ and $z_{\pi(2)}$, and so on, yielding an $O(n \log n)$
query strategy for identifying the secret. The details are 
given in Algorithm~\ref{alg:nlogn}. 

\begin{algorithm2e}[t]
 \textbf{Initialization:} 
 $x \assign (0,\ldots,0)$\;
\For{$i=1,...,n$}{
\tcp{$f(x) \ge i - 1$ and $\pi(1),\ldots,\pi(i-1)$ are already determined}
	$\pi(i) \assign {\texttt{BinSearch}}(x, i, [n] \setminus \sset{\pi(1),\ldots,\pi(i-1)})$\;
               Update $x$ by flipping $\pi(i)$\;
}\medskip
\tcp{where $\texttt{BinSearch}$ is the following function.}
$\texttt{BinSearch}(x,i,V)$\\
\lIf{$f(x) > i - 1$}{update $x$ by flipping all bits in $V$\;}
\While{$|V|> 1$}{
\tcp{$\pi(i) \in V$, $\pi(1),\ldots,\pi(i-1) \not\in V$,  and $f(x) = i-1$}
 Select a subset $F \subseteq V$ of size $|V|/2$\; 
 Create $y$ from $x$ by flipping all bits in $F$ and query $f(y)$\;
 \lIf{$f(y) = i - 1$}
 {$V \assign V\backslash F$\;}
 \lElse{$V \assign F$\;}
}
\Return the element in $V$\;
\caption{A deterministic $O(n \log n)$ strategy for the $\ourproblem$
  problem. We write $f$ instead of $f_{z,\pi}$}
\label{alg:nlogn}
\end{algorithm2e}

The \emph{lower bound} is proved by examining the decision tree of the
deterministic query scheme and exhibiting an input for which the number
of queries asked is high. 
More precisely, we show that for every deterministic strategy, there exists an input $(z,\pi)$ such that after $\Omega(n \log n)$ queries the maximal score ever returned is at most $n/2$. This is done by a simple adversarial argument:
First consider the
root node $r$ of the decision tree. 
Let $x^1$ be the first query. We proceed to the child corresponding to score 1. According
to the rules from the preceding section $V_1$ to $V_n$ are initialized to
$[n]$. Let $x$ be the next query asked by the
algorithm and let $I$ be the set of indices in which $x$ and $x^1$ agree. 
\begin{compactenum}[\hspace{\parindent}(1)]
\item If we would proceed to the child corresponding to score $0$, then $V_1$
  would become $V_1
  \setminus I$ and $V_2$ would not not change according to Rule 1. 
\item If we would proceed to the child corresponding to score $1$, then $V_1$
  would become $V_1 \cap I$ 
and $V_2$ would become $V_2 \cap I$ according to Rule 2. 
\end{compactenum}
We proceed to the child where the size of $V_1$ at most
halves. Observe that $V_1 \subseteq V_2$ always, and the maximum score is $1$. Moreover, $V_3$ to $V_n$ are not affected. 

We continue in this way until $\abs{V_1} = 2$. Let $v$ be the vertex of the
decision tree reached. Query $x^* = x^1$ is still the query with maximum
score. We choose $i_1 \in V_1$ and
$i_2 \in V_2$ arbitrarily and consider the 
subset of all inputs for which $i_1 =\pi(1)$, $i_2 =\pi(2)$,
$z_{i_1}=x^*_{i_1}$, and $z_{i_2}=1-x^*_{i_2}$. For all such inputs, the
query path followed in the decision tree descends from the root to the
node $v$. For this collection of inputs, observe that there is one
input for every assignment of values to $\pi(3),\dots,\pi(n)$
different from $i_1$ and $i_2$, and for every assignment of $0/1$ values to
$z_{\pi(3)},\dots,z_{\pi(n)}$. Hence we can recurse on this subset of
inputs starting at $v$ ignoring $V_1,V_2,\pi(1),\pi(2),z_{\pi(1)}$,
and $z_{\pi(2)}$. The setup is identical to what we started out
with at the root, with the problem size decreased by 2. 
We proceed this way, forcing $\Omega(\log n)$ queries for every two
positions revealed, until we have returned a score of $n/2$ for the first time. At this point, we have forced at least $n/4 \cdot \Omega(\log n) = \Omega(n \log n)$ queries. 
\end{proof}

\section{The Randomized Strategy}
\label{sec:upper}
\label{SEC:UPPER}

We now show that the randomized query complexity is only $O(n \log\log
n)$. The randomized strategy  overcomes the sequential learning
process of the binary search strategy (typically revealing a constant amount of
information per query) and instead has a typical information gain of
$\Theta(\log n /\log\log n)$ bits per query. In the language of the candidate
sets $V_i$, we manage to reduce the sizes of many of these sets in parallel, that is,
we gain information on several values $\pi(i)$ despite the seemingly sequential way the function 
$f_{z,\pi}$ offers information. 
The key to this is using partial information
given by the $V_i$ (that is, information that does not determine $\pi_i$, but
only restricts it) to guess with good probability an $x$ with
$f_{z,\pi}(x)>s^*$.

\begin{theorem}
\label{thm:upper}
The randomized query complexity of the $\ourproblem$ problem with $n$ positions is $O(n \log\log n)$.
\end{theorem}

The strategy has two parts. 
In the first part, we identify the positions $\pi(1), \ldots, \pi(q)$
and the corresponding bit values $z_{\pi(1)}, \ldots, z_{\pi(q)}$ for some $q
\in n - \Theta(n/\log n)$ with $O(n \log \log n)$ queries.
In the second part, we find the remaining $n-q \in \Theta(n/\log n)$ positions
and entries using the binary search algorithm with $O(\log n)$ queries per
position.

\subsection{A High Level View of the First Part}
\label{subsec:outline:main}

We give a high level view of the first part of our randomized strategy. Here and in the
following we denote by $s^*$ the current best score,  
and by $x^*$ we denote a corresponding query; i.e., $f_{z,\pi}(x^*)=s^*$. 
For brevity, we write~$f$ for~$f_{z,\pi}$.

The goal of any strategy must be to increase $s^*$ and to gain more information
about $\pi$ by reducing the sets $V_1, \ldots, V_{s^*+1}$. 
Our strategy carefully balances the two subgoals. If $\abs{\bigcup_{i \le s^*+1} V_i}$ is ``large'', it concentrates on reducing sets,
if $\abs{\bigcup_{i \le s^*+1} V_i}$ is ``small'', it concentrates on increasing $s^*$. The latter will simulatenously reduce $V_{s^*+1}$. 

We arrange the candidate sets $V_1$ to $V_n$ into
$t+2$ \emph{levels} 0 to $t+1$, where $t = \Theta(\log \log n)$. 
Initially, all candidate sets are on level $0$, and we have $V_i=[n]$ for all $i \in [n]$.
The sets in level $i$ have larger index than the sets in level
$i+1$. Level $t+1$ contains an initial segment of candidate sets, and all candidate sets on level $t+1$
are singletons, i.e., we have identified the corresponding $\pi$-value. 
On level $i$, $1 \leq i \leq t$, we can have 
up to $\alpha_i$ sets. We also say that the \emph{capacity} of level $i$ is $\alpha_i$. 
The size of any set on level $i$ is at most $n/\alpha_i^d$, where $d$ is any
constant greater than or equal to $4$. We choose $\alpha_1 = \log n$, $\alpha_i =
\alpha_{i-1}^2$ for $1 \le i \le t$ and $t$ maximal such that $\alpha_t^d \le
n$. 
Depending on the status (i.e., the fill rate) of these levels, either we try to
increase $s^*$, or we aim at reducing the sizes of the candidate sets. 

The algorithm maintains a counter $s \le s^*$ and strings $x,y \in \sset{0,1}^n$
with $f(x) = s < f(y)$. The following invariants hold for the candidate sets $V_1$ to
$V_n$:\smallskip

\begin{compactenum}[\hspace{\parindent}(1)]
\item $\pi(j) \in V_j$ for all $j$. \label{pi(j) in Vj}
\item The sets $V_j$, $j \le s$, are pairwise disjoint. \label{pairwise disjoint}
\item $V_j = [n]$ for $j > s$.
\item $V_j \setminus \{\pi(j)\}$ is random. More precisely, there is a set
  $V_j^* \subseteq [n]$ such that $\pi(j) \in V_j$ and $V_j \setminus \sset{\pi(j)}$ is a uniformly random
  subset of $V_j^* \setminus \sset{\pi(j)}$ of size
  $\abs{V_j} - 1$. \label{random subset}
\end{compactenum}\smallskip

\begin{algorithm2e}[h!]
  \textbf{Input:} 
	Number of levels $t$. 
	Capacities $\alpha_1, \ldots, \alpha_t \in \N$ of the levels $1, \ldots, t$.
	Score $q \in n - \Theta(n/\log n)$ that is to be achieved in the first phase.
	Positive integer $d \in \N$.\\
  \textbf{Main Procedure}\\
  $V_1, \ldots, V_n \assign [n]$\,; \tcp{$V_i$ is the set of candidates for $\pi(i)$}
  $s \assign 0$\,; \tcp{$s$ counts the number of successful iterations}
  $x \assign 0^n$; $y \assign 1^n$; $J \assign \emptyset$\,; 
  \lIf{$f(x) > 0$} {swap $x$ and $y$\,; \tcp{$f(x) = s < f(y)$ and $J = \sset{1,\ldots,s}$}}
  \While{$|J| < q$}{%
     \tcp{$J = [s]$, $V_j = \sset{\pi(j)}$ for $j \in J$, $f(x) = s < f(y)$, and
       $\pi(s+1) \in [n] \setminus \bigcup_{j \le s} V_j$ }
     $J' \assign \Advance(t)$\,; \tcp{$J' \not= \emptyset$}
     Reduce the size of the sets $V_j$ with $j \in J'$ to $1$ by calling
		 $\SizeReduction(\alpha_t,J',1,x)$\label{line:subroutinetwo1}\;
     $J \assign J \cup J'$\;
}
  Part 2: Identify values $\pi(n-q+1), \ldots, \pi(n)$ and
  corresponding bits using $\texttt{BinSearch}$\; \medskip
\tcp{where $\Advance$ is the following function.}
  $\Advance(\text{level } \ell)$\\
 \tcp{$\pi(s+1) \not\in \bigcup_{j=1}^s V_j$, $f(x) = s < f(y)$ and invariants (\ref{pi(j) in Vj}) to (\ref{random subset})
    hold. }
\tcp{returns a set $J$ of up to $\alpha_\ell$ indices such that $\abs{V_j} \le n/\alpha_\ell^d$ for all $j \in J$}
  $J \assign \emptyset$\; \label{line:begininnermost}
  \While{$\abs{J} \le \alpha_\ell$}{
  \tcp{$\pi(s+1) \not\in \bigcup_{j=1}^s V_j$, $f(x) = s < f(y)$ and invariants (\ref{pi(j) in Vj}) to (\ref{random subset})
    hold. $\abs{V_j} \le n/\alpha_\ell^d$ for $j \in J$.}
  \eIf{$\ell = 1$}{
   $V_{s+1}^* \assign [n] \setminus \bigcup_{j=1}^s V_j$ \; \label{line:V*}
   $V_{s+1}  \assign \RandBinSearch(x, s+1, V_{s+1}^*, n/\alpha_1^d)$\,; \tcp{Reduce
                $|V_{s+1 }|$ to $n/\alpha_1^d$} \label{querysub1}
	$s \assign s+1$\;
 		$J \assign J \cup \{s\}$\;
 		$x \assign y$\,; \tcp{establishes $f(x) \ge s$} \label{line:endinnermost}
	 	}
{ 
    $J' \assign \Advance(\ell - 1)$\,; \tcp{$J' \not= \emptyset$, and $s = \max
      J' \le f(x)$}
     Reduce the sets $V_j$, $j \in J'$, to size $n/\alpha_\ell^d$ using
     $\SizeReduction(\alpha_{\ell-1},J',n/\alpha_{\ell}^d,x)$\label{line:subroutinetwo2}\;
    $J \assign J \cup J'$\;

}

Create $y$ from $x$ by flipping all bits in $[n]\backslash
  \bigcup_{j=1}^s{V_j}$ and query $f(y)$ \;\label{query1}
 \If{($f(x) > s$ \normalfont{and} $f(y) > s$) \normalfont{or} ($f(x) = s$
   \normalfont{and} $f(y) = s$)}{{\bf break}\,; \tcp{$\pi(s+1) \in \bigcup_{j=1}^s
     V_j$; failure on level $\ell$}}\label{query1a}
   \lIf{$f(x)>s$}{swap $x$ and $y$\,; \tcp{ $\pi(s+1) \not\in \bigcup_{j=1}^s V_j$
       and $f(x) = s < f(y)$}}     
}
return $J$\;
\caption{The $O(n \log\log n)$ strategy for the $\ourproblem$ problem with $n$ positions.}
\label{alg:upper}
\end{algorithm2e} 

Our first goal is to increase $s^*$ to $\log n$ and to move the sets $V_1,
\ldots, V_{\log n}$ to the first level, i.e., to decrease their size to
$n/\alpha_1^d = n/\log^d n$. This is done sequentially. 
We
start by querying $f(x)$ and $f(y)$, where $x$ is arbitrary and $y = x \oplus
1^n$ is the bitwise complement of $x$. By swapping $x$ and $y$ if needed, we may assume $f(x) = 0 < f(y)$. We
now run a randomized binary search for finding $\pi(1)$. We choose uniformly at
random a subset $F_1 \subseteq V_1$ ($V_1=[n]$ in the beginning) of size $|F_1| = |V_1|/2$. 
We query $f(y')$ where $y'$ is obtained from $x$ by flipping the bits in
$F_1$. 
If $f(y') = 0$, we set $V_1 \assign V_1\setminus F_1$; we set $V_1 \assign F_1$
otherwise. This ensures $\pi(1) \in V_1$
and invariant (\ref{random subset}). 
We stop this binary search once $\pi(2)
\not\in V_1$ is sufficiently likely; the analysis will show that $\Pr[\pi(2)
\in V_1] \leq 1/\log^d n$ (and hence $|V_1| \leq n/\log^d n$) for some large
enough constant $d$ is a good choice.

We next try to increase $s$ to a value larger than one and to simultaneously decrease the size of
$V_2$. Let $\{x,y\} = \{y,y \oplus 1_{[n]\setminus V_1}\}$. If $\pi(2) \not\in V_1$, one of
$f(x)$ and $f(y)$ is one and the other is larger than one. 
Swapping $x$ and
$y$ if necessary, we may assume $f(x) = 1 < f(y)$. We use randomized
binary search to reduce the size of $V_2$ to $n/\log^d n$. The randomized
binary search is similar to before. 
Initially, $V_2$ is equal to $V_2^* = [n] - V_1$. 
At each step we chose a subset $F_2 \subseteq V_2$ of size $|V_2|/2$ and we create $y'$ from $x$ by flipping the bits in positions $F_2$. 
If $f(y') > 1$ we update $V_2$ to $F_2$ and we update $V_2$ to $V_2 \setminus F_2$ otherwise. We stop once $|V_2| \leq n/\log^d n$.

At this point we have $|V_1|, |V_2| \leq n/\log^d n$ and $V_1 \cap V_2 = \emptyset$. 
  We hope that $\pi(3) \notin V_1 \cup V_2$, in which case we can increase
  $s$ to three and move set $V_3$ from level $0$ to level $1$ by random
  binary search 
  (the case $\pi(3) \in V_1 \cup V_2$ is called a \emph{failure} and will be treated separately at the end of this overview).

At some point the probability that $\pi(i) \notin V_1 \cup \ldots \cup V_{i-1}$
drops below a certain threshold and we cannot ensure to make progress anymore
by simply querying $y \oplus ([n]\backslash (V_1 \cup \ldots \cup
V_{i-1}))$. 
This situation is reached when $i = \log n$ and
  hence we
abandon the previously described strategy once $s= \log n$. At
this point, we move our focus from increasing $s$ to
reducing the size of the candidate sets $V_1, \ldots, V_{s}$, thus adding
them to the second level. More precisely, we reduce their sizes to at most
$n/\log^{2d} n = n/\alpha_2^d$. This reduction is carried out by
$\SizeReduction$, which we describe in Section~\ref{subsec:subroutine2}.  
It reduces the sizes of the up to $\alpha_{\ell - 1}$
  candidate sets from 
  some value $\leq n/\alpha_{\ell - 1}^d$ to 
  the target size $n/{\alpha_{\ell}^d}$ of level $\ell$ with 
an expected number of 
$O(1)\alpha_{\ell - 1} d (\log (\alpha_{\ell})-\log (\alpha_{\ell -
   1}))/\log (\alpha_{\ell-1})$ queries. 

Once the sizes $|V_1|, \ldots, |V_{s}|$ have been reduced to at most
$n/\log^{2d} n$, we move our focus back to increasing $s$. The probability
that $\pi(s+1) \in V_1 \cup \ldots \cup V_{s}$ will now be small enough 
(details below), and we proceed as before by
flipping in $x$ the entries in the positions $[n] \setminus (V_1 \cup \ldots \cup V_{s})$ and reducing the size
of $V_{s+1}$ to $n/\log^d n$.  
Again we iterate this process until the first level is filled; i.e., until we have $s=2 \log n$. 
As we did with $V_1, \ldots, V_{\log n}$, we reduce the sizes of 
$V_{\log n+1}, \ldots, V_{2\log n}$ to $n/\log^{2d} n = n/\alpha_2^d$, thus adding them to the second level.
We iterate this process of moving $\log n$ sets from level $0$ to level $1$ and
then moving them to the second level until $\log^2 n = \alpha_2$ sets have been added to the second level. 
At this point the second level has reached its capacity and we proceed by
reducing the sizes of $V_1, \ldots, V_{\log^2 n}$ to at most $n/\log^{4d} n = n/\alpha_3^d$, thus adding them to the \emph{third level.}

In total we have $t=O(\log \log n)$ levels. 
For $1 \leq i \leq t$, the $i$th level has a capacity of 
$\alpha_i:= \log^{2^{i-1}} n$ sets, each of which is required to be of size at most
$n/\alpha_i^d$.
Once level $i$ has reached its capacity, we reduce the size of the sets on the $i$th level to at most 
$n/\alpha_{i+1}^d$, thus moving them from level $i$ to level $i+1$. 
When $\alpha_t$ sets $V_{i+1}, \ldots, V_{i+\alpha_t}$ have been added to the last level, level $t$, we finally reduce their sizes to one.
This corresponds to determining 
$\pi(i+j)$ for each $j \in [\alpha_t]$.

\paragraph{Failures:} 
We say that a failure happens if we want to move some set $V_{s+1}$ from level 0 to level 1, but
$\pi(s+1) \in V_1 \cup \ldots \cup V_s$. In case of a failure, we immediately
stop our attempt of increasing $s$. Rather, we \emph{abort} the first level and
move all sets on the first level to the second one. 
As before, this is done by calls to $\SizeReduction$ which reduce the size of the sets from at most $n/\log^d n$ to at most $n/\log^{2d} n$.
We test whether we now have $\pi(s+1) \not\in V_1 \cup \ldots \cup V_s$.
Should we still have $\pi(s+1) \in V_1 \cup \ldots \cup V_s$, we continue by
moving all level $2$ sets to level $3$, and so on, until we finally have
$\pi(s+1) \not\in V_1 \cup \ldots \cup V_s$. 
At this point, we proceed again by moving sets from level $0$ to level $1$,
starting of course with set $V_{s+1}$. The condition $\pi(s+1) \not\in V_1 \cup
\ldots \cup V_s$
will certainly be fulfilled once we have moved $V_1$ to $V_s$ to level $t+1$,
i.e., have reduced them to singletons.

\paragraph{Part 2:} In the second part of Algorithm~\ref{alg:upper} we
determine the last $\Theta(n/\log n)$ entries of~$z$ and $\pi$.  This can be done as follows.
When we leave the first phase of Algorithm~\ref{alg:upper}, we have
$|V_1|=\ldots=|V_q|=1$ and $f(x)\geq q$. We can now proceed as in deterministic
algorithm (Algorithm~\ref{alg:nlogn}) and identify each of the remaining entries with
$O(\log n)$ queries. Thus the total number of queries in Part 2 is linear. 

Our strategy is formalized by Algorithm~\ref{alg:upper}. In what follows, we
first present the two subroutines, $\subroutineone$ and $\SizeReduction$. In
Section~\ref{subsec:proofupper}, we present the full proof of
Theorem~\ref{thm:upper}.

\subsection{Random Binary Search}
\label{subsec:subroutine1}

$\RandBinSearch$ is called by the function $\Advance(1)$. It reduces the size
of a candidate set from some value $v \le n$ to some value $\ell<v$ in $\log v - \log \ell$ queries.  

\begin{algorithm2e}[h!]
\textbf{Input:}
 A position $s$, a string $x \in \{0,1\}^n$ with $f(x) = s$, and a set $V$
 with $\pi(s+1) \in V$ and $\pi(1),\ldots,\pi(s) \not\in V$, and a target size $\ell \in \N$.\\
\While{$|V|>\ell$}{
\tcp{$\pi(s+1) \in V$, $\pi(1),\ldots,\pi(s) \not\in V$,  and $f(x) = s$}
 Uniformly at random select a subset $F \subseteq V$ of size $|V|/2$\; \label{line:uniform}
 Create $y'$ from $x$ by flipping all bits in $F$ and query $f(y')$\;\label{lineyprime1}
 \lIf{$f(y') = s$}
 {$V \assign V\backslash F$\;}
 \lElse{$V \assign F$\;}
}
\textbf{Output:} Set $V$ of size at most $\ell$. 
\caption{A call $\subroutineone(x,s+1,V, \ell)$ reduces the size of the candidate
  set $V$ for $V_{s+1}$ from $v$ to $\ell$ in $\log v - \log \ell$ queries.}
\label{alg:subroutine1}
\end{algorithm2e} 

\begin{lemma}
\label{lem:binary}
Let $x \in \{0,1\}^n$ with $f(x) = s$ and let $V$ be any set with $\pi(s+1) \in
V$ and $\pi(j) \not\in V$ for $j \le s$. Let $v:=|V|$ and $\ell \in \N$ with $\ell < v$.
Algorithm~\ref{alg:subroutine1} reduces the size of $V$ to $\ell$ using at most $\lceil \log v - \log \ell \rceil$ queries.
\end{lemma}
\begin{proof} Since $f(x) = s$, we have $x_{\pi(i)} =z_{\pi(i)}$
  for all $i \in [s]$ and $x_{\pi(s+1)} \not= z_{\pi(s+1)}$. Also $\pi(s+1) \in V$ and $\pi(j) \not\in V$ for $j \le
  s$. 
Therefore, either we have $f(y')>s$ in line~\ref{lineyprime1} or we have
$f(y')=s$. In the former case, the bit $\pi(s+1)$ was flipped, and hence
$\pi(s +1) \in F$ must hold. 
In the latter case the bit in position $\pi(s +1)$ bit was not flipped and we infer $\pi(s+1) \not\in F$.

The runtime bound follows from  the fact that the size of the set $V$ halves in each iteration.
\end{proof}




We call $\subroutineone$ in Algorithm~\ref{alg:upper} (line~\ref{querysub1}) to reduce the size of $V_{s+1}$ to $n/\alpha_1^d$, or, put differently, to reduce the number of candidates for $\pi(s+1)$ to $n/\alpha_1^d$.
As the initial size of $V_{s+1}$ is at most $n$, this requires at most $d \log \alpha_1$ queries by Lemma~\ref{lem:binary}.

\begin{lemma}
\label{lem:queries1}
A call of $\Advance(1)$ in Algorithm~\ref{alg:upper} requires at most $\alpha_1 + \alpha_1 d \log \alpha_1$ queries. 
\end{lemma}

\begin{proof}
The two occasions where queries are made in Algorithm~\ref{alg:upper} are in line~\ref{query1} and in line~\ref{querysub1}. Line~\ref{query1} is executed at most $\alpha_1$ times, each time causing exactly one query.
Each call to $\subroutineone$ in
line~\ref{querysub1} causes at most $d \log \alpha_1$ queries, and
$\subroutineone$ is called at most $\alpha_1$ times.
\end{proof}

\subsection{Size Reduction}
\label{subsec:subroutine2}

We describe the second subroutine of Algorithm~\ref{alg:upper}, $\SizeReduction$. This routine
  is used to reduce the sizes of the up to $\alpha_{\ell - 1}$
  candidate sets returned by a recursive call $\Advance(\ell - 1)$ from 
  some value $\leq n/\alpha_{\ell - 1}^d$ to 
   at most the target size of level $\ell$, which is $n/{\alpha_{\ell}^d}$.
   As we shall see below, this requires an expected number of $O(1)\alpha_{\ell - 1} d (\log \alpha_{\ell}-\log \alpha_{\ell -
   1})/\log \alpha_{\ell-1}$ queries. 
The pseudo-code of $\SizeReduction$ is given in
Algorithm~\ref{alg:SizeReduction}. It repeatedly calls a subroutine
$\ReductionStep$ that reduces the 
sizes of at most $k$ candidate sets to a $k$th fraction of their original size
using at most $O(k)$ queries, where $k$ is a parameter. We use $\ReductionStep$
with parameter $k = \alpha_{\ell - 1}$ repeatedly to achieve the full reduction
of the sizes to at most $n/\alpha_{\ell}^d$.

\begin{algorithm2e}[h]
\textbf{Input:} Positive integer $k \in \N$, a set $J \subseteq [n]$ with
$\abs{J} \le k$, $s = \max{J}$, a 
target size $m \in \N$, and a string $x \in \{0,1\}^n$ such that $f(x)\geq \max J$, and invariants (\ref{pi(j) in Vj}) to (\ref{random
  subset}) hold\\	
Let $\alpha$ be such that $n/\alpha^d = \max_{j \in J} \abs{V_j}$ and let 
$\beta$ be such that $n/\beta^d = m$\,;\\
\For{$i = 1,\ldots,d(\log \beta - \log \alpha)/\log k$}
{ $\ReductionStep(k,J,n/(k^i \alpha^d),x)$\,;\\ }
\textbf{Output:} Sets $V_j$ with $|V_{j}| \leq m$ for all $j \in J$.
\caption{A call $\SizeReduction(k, J,m, x)$ reduces
  the size of at most $k$ sets $V_j$, $j\in J$, to size at most $m$. We use it
  twice in our main strategy. In line~\ref{line:subroutinetwo2}, we call 
$\SizeReduction(\alpha_{\ell-1}, J', n/\alpha_{\ell}^d, x)$ to reduce the size
of each $V_j$, $j \in J'$ to $n/\alpha_{\ell}^d$.  In line~\ref{line:subroutinetwo1}, we call
$\SizeReduction(\alpha_t, J', 1, x)$ to reduce the size of each $V_j$, $j \in
J'$, to one.}
\label{alg:SizeReduction}
\end{algorithm2e}

$\ReductionStep$ is given a set $J$ of at most $k$ indices and a string $x$ with
  $f(x)\geq \max J$. The goal is to reduce the size
of each candidate set $V_j$, $j \in J$, below a target size $m$ where $m \ge
\abs{V_j}/k$ for all $j \in J$. The routine works in phases of several iterations each. 
Let $J$ be the set
of indices of the candidate sets that are still above the target size at the
beginning of an iteration. For each $j \in J$, we 
randomly choose a subset $F_{j} \subseteq V_{j}$ of size $|V_{j}|/k$. 
We create a new bit-string $y'$ from $x$ 
by flipping the entries in positions $\bigcup_{j \in J}F_{j}$. Since the sets
$V_j$, $j \le s = \max J$, are pairwise disjoint, we have either 
$f(y') \geq \max J$ or 
$f(y')= j - 1$ for some $j \in J$.
%
In the first case, i.e., if $f(y') \geq \max J$, none of the sets $V_{j}$ was
  \emph{hit}, and for all $j \in J$ we can remove the subset $F_j$ from the candidate set $V_j$.
We call such queries ``\emph{off-trials}''.
An off-trial reduces the size of all sets $V_{j}$, $j \in J$, to a
$(1-1/k)$th fraction of their original size.
If, on the other hand, we have $f(y')=j-1$ for some $j \in J$, we can replace
  $V_j$ by set $F_j$ as $\pi(j) \in F_{j}$ must hold.
Since $|F_{j}| = |V_{j}|/k \leq m$ by assumption, this set has now been reduced
to its target size and we can remove it from $J$.

We continue in this way until at least half of the indices are removed
  from $J$ and at least $ck$ off-trials occurred, for some constant $c$ satisfying $(1-1/k)^{ck} \leq 1/2$. 
  We then proceed to the next phase.
  Consider any $j$ that is still
  in $J$. The size of $V_j$ was reduced by a factor $(1 - 1/k)$ at least $ck$
  times. Thus its size was reduced to at most half its original size. 
  We may thus halve $k$ without destroying the invariant $m \ge \abs{V_j}/k$ for $j \in
  J$. The effect of halving $k$ is that the relative size of the sets $F_j$
  will be doubled for the sets $V_j$ that still take part in the reduction process.

\begin{algorithm2e}[h!]
\textbf{Input:} Positive integer $k \in \N$, a set $J \subseteq [n]$ with
$\abs{J} \le k$, a target size $m \in \N$ with $|V_j| \leq k m$ for all $j
\in J$, and a string $x \in \{0,1\}^n$ with $f(x)\geq \max J$. Invariants
(\ref{pi(j) in Vj}) to (\ref{random subset}) hold.\\	
 \lFor{$j \in J$}
 {\lIf{$|V_j |\leq m$}{delete $j$ from $J$\,; // $V_{j}$ is already small enough}}\\
\While{$J \not= \emptyset$}{
 $o \assign 0$\,; // counts the number of off-trials\\
 $\ell = \abs{J}$\,; // $|V_j| \leq k m$ for all $j \in J$\\
\Repeat{$o\geq c\cdot k$ and $\abs{J} \le \ell/2$ \normalfont{// $c$ is chosen such that $(1 - 1/k)^{ck} \le 1/2$}\label{line:condition2}}{
 \lFor{$j \in J$}
 {Uniformly at random choose a subset $F_{j} \subseteq V_{j}$ of size $|V_{j}|/k$\;}
 Create $y'$ from $x$ by flipping in $x$ the entries in positions $\bigcup_{j \in J}{F_{j}}$ and query $f(y')$\; \label{line:yprime}
 \eIf{$f(y')\geq \max J$}{
 			$o \assign o+1$\,; // ``off''-trial\\
 			\lFor{$j\in J$}{$V_{j} \assign V_{j} \backslash F_{j}$\;}
 			}
 			{
 			$V_{f(y')+1} \assign F_{f(y')+1}$\,; // set $V_{f(y')+1}$ is hit\\
 			\lFor{$j \in J $} \lIf{$j
                          \le f(y')$}{{$V_{j} \assign V_{j} \backslash F_{j}$\;}
 			}}
  \lFor{$j \in J$}
 {\lIf{$|V_{j}|\leq m$}{delete $j$ from $J$\;} }
}
$k \assign k/2$\;
}
\textbf{Output:} Sets $V_j$ with $|V_{j}| \leq m$ for all $j \in J$.
\caption{Calling $\ReductionStep(k, J,m, x)$ reduces the size of at most $k$ sets $V_j$, $j\in J$, to
    a $k$th fraction of their original size using only $O(k)$ queries.}
\label{alg:ReductionStep}
\end{algorithm2e} 

\begin{lemma}
\label{lem:reduce}
Let $k \in \N$ and let $J \subseteq [n]$ be a set of at most $k$ indices with
$s = \max J$. Assume that invariants (\ref{pi(j) in Vj}) and (\ref{random
  subset}) hold. 
Let $x \in \{0,1\}^n$ be such that $f(x)\geq
\max J$ and let $m \in \N$ be such that 
$m\geq  |V_j|/k$ for all $j \in J$.
In expectation it takes $O(k)$ queries until 
$\normalfont{\ReductionStep}(k,J,m,x)$ has reduced the size of $V_{j}$ to at most $m$
for each $j \in J$. 
\end{lemma}

\begin{proof}
Let $c$ be some constant.
We show below that---for a suitable choice of $c$---after an expected number of
at most $ck$ queries both conditions in line~\ref{line:condition2} are
satisfied.  
Assuming this to hold, we can bound the total expected number of queries
until the size of each of the sets $V_{j}$ has been reduced to $m$ by 
\begin{align*}
\sum_{h=0}^{\log k}{c k/2^h} < 2ck\,,
\end{align*}
as desired.

In each iteration of the repeat-loop we either hit an index in $J$ and hence
remove it from $J$ or we have an off-trial. The probability of an off-trial is at
least $(1 - 1/k)^k$ since $\abs{J} \le k$ always. Thus the probability of an
off-trial is at least $(2e)^{-1}$ and hence the condition $o \ge ck$ holds after
an expected number of $O(k)$ iterations. 

As long as $\abs{J} \ge \ell/2$, the probability of an off-trial is at most $(1
- 1/k)^{\ell/2}$ and hence the probability that a set is hit is at least $1 - (1
- 1/k)^{\ell/2}$. Since $\ln(1 - 1/k) \le -1/k$ we have $(1 - 1/k)^{\ell/2} =
\exp(\ell/2 \ln(1 - 1/k)) \le \exp(-\ell/(2k))$ and hence $1 - (1
- 1/k)^{\ell/2} \ge 1 - \exp(-\ell/(2k)) \ge \ell/(2k)$. Thus the expected number
of iterations to achieve $\ell/2$ hits is $O(k)$. 

If a candidate set $V_j$ is hit in the repeat-loop, its size is reduced to
$\abs{V_j}/k$. By assumption, this is bounded by $m$. If $V_j$ is never hit,
its size is reduced at least $ck$ times by a factor $(1 - 1/k)$. By choice of
$c$, its size at the end of the phase is therefore at most half of its original size. Thus after replacing $k$ by $k/2$
we still have $\abs{V_j}/k \le m$ for $j \in J$. 
\end{proof}

It is now easy to determine the complexity of $\SizeReduction$.
\begin{corollary}
\label{cor:reduce}
Let $k \in \N$, $J$, and $y$ be as in Lemma~\ref{lem:reduce}. 
Let further $d \in \N$ and $\alpha \in \R$ such that 
$\max_{j \in J} |V_{j}| = n/\alpha^d$.
Let $\beta \in \R$ with $\beta > \alpha$. 

Using at most $d (\log \beta - \log \alpha)/\log k$ calls to
Algorithm~\ref{alg:ReductionStep} we can reduce the maximal size $\max_{j \in J} |V_{j}|$ to $n/\beta^d$. 
The overall expected number of queries needed by $\SizeReduction$ to achieve this reduction is 
$O(1) k d (\log \beta - \log \alpha) /\log k$. 
\end{corollary}

\begin{proof}
The successive calls can be done as follows. We first call $\ReductionStep(k,J,
n/(k \alpha^d), x)$. By Lemma~\ref{lem:reduce} it takes an expected number of $O(k)$
queries until the algorithm terminates. The sets $V_j$, $j \in J$, now have 
size at most $n/(k \alpha^d)$. 
We next call $\ReductionStep(k,J,n/(k^2 \alpha^d), x)$. 
After the $h$th such call we are left with sets of size at most $n/(k^{h} \alpha^d)$. For $h= d (\log \beta - \log \alpha) /\log k$ we have 
$k^h \geq (y/\alpha)^d$. 
The total expected number of queries at this point is $O(1) k d (\log \beta - \log \alpha) /\log k$.
\end{proof}

\subsection{Proof of Theorem~\ref{thm:upper}}
\label{subsec:proofupper}

It remains to show that the first phase of Algorithm~\ref{alg:upper} takes $O(n
\log\log n)$ queries in expectation. 
 
\begin{theorem}
\label{thm:phase1}
Let $q\in n - \Theta(n/\log n)$.
Algorithm~\ref{alg:upper} identifies positions $\pi(1), \ldots,
\pi(q)$ and the corresponding entries $z_{\pi(1)}, \ldots, z_{\pi(q)}$ of $z$
in these positions using an expected number of $O(n \log \log n)$ queries.
\end{theorem}

We recall from the beginning of Section~\ref{sec:upper} that Theorem~\ref{thm:upper} follows from Theorem~\ref{thm:phase1} by observing that once the $q\in n - \Theta(n/\log n)$ first positions of $z$ have been identified, the remaining $\Theta(n/\log n)$ values $z_{\pi(q+1)}, \ldots, z_{\pi(n)}$ can be identified by binary search in expected linear time (already a na\"ive implementation of binary search finds each $z_{\pi(j)}$ in $O(\log n)$ queries, as demonstrated in Section~\ref{sec:deterministic}).

To prove Theorem~\ref{thm:phase1}, we bound by $O(n \log \log n)$ the expected number of queries needed if no failure in any call of $\Advance$ happened and by $O(n)$ the expected number of queries caused by failures (the latter is postponed to Section~\ref{sec:failure}).   
	
If no failure in any call of $\Advance$ happened, the expected number of
queries is bounded by
\begin{align}
\nonumber
& \frac{q}{\alpha_t}
\bigg( 
\frac{\alpha_t \log n}{\log \alpha_{t}} + \frac{\alpha_{t}}{\alpha_{t-1}}
\bigg( 
 \frac{\alpha_{t-1} d c (\log \alpha_{t} - \log \alpha_{t-1})}{\log \alpha_{t-1}} \\
&\quad\quad\quad\quad\quad\quad\quad\quad\quad + 
\frac{\alpha_{t-1}}{\alpha_{t-2}}
\bigg(
\ldots
 + \frac{\alpha_{2}}{\alpha_{1}}
 \bigg(
\frac{\alpha_1 d c (\log \alpha_2- \log \alpha_1)}{\log \alpha_{1}} + \alpha_1 d \log \alpha_1
\bigg)
\bigg)
\bigg)
\bigg)\nonumber \\
\leq  
& n d c \left(
\frac{\log n}{\log \alpha_{t}} + \frac{\log \alpha_{t}}{\log \alpha_{t-1}} 
 + \ldots + \frac{\log \alpha_2}{\log \alpha_{1}} + \log \alpha_1 -t
\right)
\,,
\label{eq:overall1:main}
\end{align}
where $c$ is the constant hidden in the $O(1)$-term in Lemma~\ref{lem:reduce}. 
To verify this formula,
observe that 
we fill the $(i-1)$st level $\alpha_i/\alpha_{i-1}$ times before level $i$ has reached its capacity of $\alpha_i$ candidate sets. 
To add $\alpha_{i-1}$ candidate sets from level $i-1$ to level $i$, we 
need to reduce their sizes from $n/\alpha_{i-1}^d$ to $n/\alpha_{i}^d$. 
By Corollary~\ref{cor:reduce} this requires at most $\alpha_{i-1} d c (\log \alpha_{i}-\log \alpha_{i-1})/\log \alpha_{i-1}$ queries.
The additional $\alpha_1 d \log \alpha_1$ term 
accounts for the queries 
needed to move the sets from level $0$ to level $1$; 
i.e., for the randomized binary search algorithm through which we initially reduce the sizes of the sets 
$V_i$ to
$n/\alpha_1^d$---requiring at most $d \log \alpha_1$ queries per call.
Finally, the term $\alpha_t \log n /\log \alpha_{t}$ accounts for the final reduction of the sets $V_i$ to a set containing only one single element (at this stage we shall finally have $V_i=\{\pi(i)\}$). More precisely, this term is $\left(\alpha_t (\log n - d \log \alpha_t)\right)/\log \alpha_{t}$ but we settle for upper bounding this expression by the term given in the formula.

Next we need to bound the number of queries caused by \emph{failures}. 
We show that, on average, not too many failures happen. 
More precisely, we will show in Corollary~\ref{cor:failure} that the expected number of level-$i$ failures is at most $n^2/((n-q)(\alpha_i^{d-1}-1))$. 
By Corollary~\ref{cor:reduce}, each such level-$i$ failure causes an additional
number of at most $1 + \alpha_i d c (\log \alpha_{i+1} -\log \alpha_i)/\log \alpha_i $ queries 
(the $1$ counts for the query through which we discover that 
$\pi(s+1) \in V_1 \cup \ldots \cup V_{s}$).
Thus,
\begin{align}
\label{eq:overallfailure:main}
	\sum_{i=1}^t{\frac{n^2}{(n-q)(\alpha_i^{d-1}-1)}}\left(1 + \frac{\alpha_i d c (\log \alpha_{i+1} -\log \alpha_i)}{\log \alpha_i}\right)
\end{align}
bounds the expected number of additional queries caused by failures. 

Recall the choice of $t$ and the capacities $\alpha_i$.
We have $\alpha_1=\log n$ and $\alpha_i = \alpha_{i-1}^2 = (\log
n)^{2^{i-1}}$; $t$ is maximal such that $\alpha_t^d \le n$. Then $\alpha_t
\ge \sqrt{n}$ and hence $\log \alpha_t = \Omega(\log n)$. 
The parameter $d$ is any constant that is at least $4$. 
With these parameter settings, formula (\ref{eq:overall1:main}) evaluates to 
\begin{align*}
	n d c \left(\frac{\log n}{\log \alpha_t} + 2(t-1) + \log \log n - t\right) 
	= O(n \log \log n)
\end{align*}
and, somewhat wastefully, we can bound formula (\ref{eq:overallfailure:main}) from above by
\begin{align*}
	\frac{n^2 d c}{n-q} \sum_{i=1}^t{\alpha_i^{-(d-3)}}
	= O(n \log n) 
			\sum_{i=0}^{t-1}{\alpha_1^{-(d-3)2^i}}
	< O(n \log n) (\alpha_1^{d-3} -1)^{-1} 
	=O(n)
	\,,
\end{align*}
where the first equation is by construction of the values of $\alpha_i$, the inequality
uses the fact that the geometric sum is dominated by the first term, and the
last equality stems from our choice $\alpha_1=\log n$ and $d \ge 4$.
This shows, once Corollary~\ref{cor:failure} is proven, that the overall expected number of queries sums to 
$O(n \log \log n)+O(n)= O(n \log \log n)$. \qed

\subsection{Failures}\label{sec:failure}

We derive a bound on the expected number of failures. We do so in two
steps. We first bound the worst-case number of calls of $\Advance(\ell)$ for
any fixed $\ell$ and
then bound the probability that any particular call may
fail. 

In order to bound the worst-case number of calls of $\Advance(\ell)$, we
observe that any such call returns a non-empty set and distinct calls return
disjoint sets. Thus there can be at most $q$ calls to $\Advance(\ell)$ for any
$\ell$. Before a call to $\Advance(\ell)$, we have $\pi(s+1) \not\in
\bigcup_{j \le s} V_j$, and hence any call increases $s$ by at least 1. This is
obvious for $\Advance(1)$ and holds for $\ell > 1$, since $\Advance(\ell)$
immediately calls $\Advance(\ell - 1)$. 

We turn to the probabilistic part of the analysis. Key to the failure analysis
is the following observation.

\begin{lemma} 
\label{lem:random2}
Each $V_j$ has the property that $V_j \setminus \{\pi(j)\}$ is
  a random subset of $V_j^* \setminus \{\pi(j)\}$ of size $\abs{V_j} - 1$,
  where $V_j^*$ is as defined in line~\ref{line:V*} of
  Algorithm~\ref{alg:upper}. 
\end{lemma}

\begin{proof}
$V_j$ is initialized to $V_j^*$; thus the claim is true
  initially. In $\subroutineone$, a random subset $F$ of $V_j$ is chosen and $V_j$
  is reduced to $F$ (if $\pi(j) \in F)$ or to $V_j \setminus F$ (if $\pi(j)
  \not\in F)$. In either case, the claim stays true. The same reasoning
  applies to $\ReductionStep$. 
\end{proof}

\begin{lemma}
\label{lem:failure}
Let $q \in n-\Theta(n/\log n)$ be the number of indices $i$ for which we
determine $\pi(i)$ and $z_{\pi(i)}$ in the first phase.  
The probability that any particular call of
  $\Advance(\ell)$ fails is less than $n (n-q)^{-1} (\alpha_\ell^{d-1} - 1)^{-1}$.
\end{lemma}

\begin{proof}
A failure happens if both $x$ and $y = x \oplus 1_{[n]\setminus
  \bigcup_{i=1}^{s}{V_i}}$ have a score of more than $s$ or equal to $s$ in
line~\ref{query1a} of Algorithm~\ref{alg:upper}.  
This is equivalent to $\pi(s+1) \in \bigcup_{i=1}^{s}{V_i}$. 

Let $k$ be the number of indices $i \in [n]$ for which $V_i$ is on the last level; i.e., $k:=|\{ i \in [n] \mid |V_i|=1\}|$ is the number of sets $V_i$ which have been reduced to singletons already. Note that these sets satisfy $V_i=\{\pi(i)\}$. 
Therefore, they cannot contain $\pi(s+1)$ and we do not need to take them into account.

A failure on level $\ell$ occurs only if $\pi(s+1) \in \bigcup_{j=1}^{s}{V_j}$ and the size of each candidate set $V_1, \ldots, V_s$ has been reduced already to at most $n/\alpha_\ell^d$. 
There are at most $\alpha_j$ candidate sets on each level $j\geq \ell$. 
By construction, the size of each candidate set on level $j$ is at most $n/\alpha_j^d$. 
By Lemma~\ref{lem:random2}, the probability that $\pi(s+1) \in \bigcup_{j=1}^{s}{V_j}$ is at most 
\begin{align}
\label{eq:1}
\frac{1}{n-k} \sum_{j=\ell}^t{\frac{n\alpha_j}{\alpha_j^{d}}}
=
\frac{n}{n-k} \sum_{j=\ell}^t{\frac{1}{\alpha_j^{d-1}}}\,.
\end{align}

By definition we have $\alpha_j \geq \alpha_\ell^{(2^{j-\ell})}$ and in particular we have $\alpha_j \geq \alpha_\ell^{{j-\ell}}$. Therefore expression (\ref{eq:1}) can be bounded from above by 
\[
\frac{n}{n-k}
\sum_{j=1}^{t-\ell}{\left(\frac{1}{\alpha_\ell^{d-1}}\right)^j} 
<
\frac{n}{n-k}
\left( \alpha_\ell^{d-1}-1\right)^{-1}
\,.\]
\end{proof}

Since $\Advance(\ell)$ is called at most $q$ times we immediately get the following.
\begin{corollary}
\label{cor:failure} Let $\ell \in [t]$. 
The expected number of level $\ell$ failures is less than 
$$
q n (n-q)^{-1} (\alpha_\ell^{d-1}-1)^{-1} \leq n^2 (n-q)^{-1} \left(\alpha_\ell^{d-1}-1\right)^{-1}
\,.
$$
\end{corollary}

\section{The Lower Bound}
\label{sec:lower}
\label{SEC:LOWER}

We prove a tight lower bound for the randomized query complexity of the $\ourproblem$ problem. The lower bound is stated in the following:

\begin{theorem}
\label{thm:lower}
The randomized query complexity of the $\ourproblem$ problem with $n$ positions is $\Omega(n \log\log n)$.
\end{theorem}

To prove a lower bound for randomized query schemes, we appeal to
Yao's principle. That is, we first define a hard distribution over the
secrets and show that every deterministic query scheme for this hard
distribution needs $\Omega(n \log \log n)$ queries in
expectation. This part of the proof is done using a potential function
argument.

\paragraph{Hard Distribution.}
Let $\Pi$ be a permutation drawn uniformly among all the permutations of
$[n]$ (in this section, we use capital letters to denote random variables). 
Given such a permutation, we let our target string $Z$ be the one
satisfying $Z_{\Pi(i)} = i (\mod 2)$ for $i=1,\dots,n$.
Since $Z$ is uniquely determined by the permutation $\Pi$,
we will mostly ignore the role of $Z$ in the rest of this section.
Finally, we use $F(x)$ to denote the value of the random variable
$f_{Z,\Pi}(x)$ for $x \in \left\{ 0,1 \right\}^n$.
We will also use the notation $a \equiv b$ to mean that $a \equiv b (\mod 2)$.

\paragraph{Deterministic Query Schemes.}
By fixing the random coins, a randomized solution with expected $t$
queries implies the existence of a deterministic query scheme with
expected $t$ queries over our hard distribution. The rest of this
section is devoted to lower bounding $t$ for such a deterministic
query scheme.

Recall that a deterministic query scheme is a decision tree $T$
in which each node $v$ is labeled with a string $x_v \in
\{0,1\}^n$. Each node has $n+1$ children, numbered from 0 to $n$, 
and the $i$th child is traversed if $F(x_v) = i$.
To guarantee correctness, no two inputs can end up in the same leaf.

For a node $v$ in the decision tree $T$, we define $\max_v$ as the largest value of $F$
seen along the edges from the root to $v$. 
Note that $\max_v$ is not a random variable.
Also, we define $S_v$ as the subset of inputs (as outlined above) that reach node $v$.

\paragraph{Update Rules for Candidate Sets.} We use a potential function which
measures how much ``information'' the queries asked have revealed
about $\Pi$. Our goal is to show that the expected increase in the potential
function after asking each query is small. Our potential function depends crucially on the
candidate sets.  The update rules for the
candidate sets are slightly more specific than the ones in Section~\ref{sec:definitions} because we
now have a fixed connection between the two parts of the secret. We denote the
candidate set for $\pi(i)$ at node $v$ with $V_i^v$. At the root node $r$, we have $V^r_i = [n]$ for all $i$. 
Let $v$ be a node in the tree and let
$w_0,\dots,w_n$ be its children ($w_i$ is traversed when the score $i$ is returned).
Let $P_0^v$ (resp. $P_1^v$) be the set of positions in $x_v$ that contain
0 (resp. 1). Thus, formally, 
$P_0^v = \{ i \mid x_v[i]=0\}$ and $P_1^v = \{ i \mid x_v[i]=1\}$.\footnote{To prevent our notations from becoming too overloaded, here and in the remainder of the section we write $x=(x[1],\ldots, x[n])$ instead of $x=(x_1,\ldots, x_n)$} 
The precise definition of candidate sets is as follows:
\[
V^{w_j}_i = \left\{
\begin{array}{l l}
  V^v_i \cap P^v_{i (\mod 2)} & \quad \text{if } i \leq j\,,\\
  V^v_i \cap P^v_{j (\mod 2)} & \quad \text{if } i = j+1\,,\\
  V^v_i & \quad \text{if } i > j+1.\\
\end{array} \right.
\]

As with the upper bound case, the candidate sets have some very useful properties.
These properties are slightly different from the ones observed before, due to the fact
that some extra information has been announced to the query algorithm.
We say that a candidate set $V_i^v$ is {\emph{active (at $v$)}} if the following
conditions are met: 
(i) at some ancestor node $u$ of $v$, we
have $F(x_u) = i-1$, 
(ii) at every ancestor node $w$ of $u$
we have $F(x_w) < i-1$, 
and (iii) $i < \min \left\{  n/3,\max_v\right\}$.
We call $V_{\max_v+1}^v$ {\emph{pseudo-active (at $v$)}}. 

For intuition on the requirement $i < n/3$, observe from the following
lemma that $V^v_{\max_v+1}$ contains all sets $V^v_i$ for $i \leq
\max_v$ and $i \equiv \max_v$. At a high level, this means that the
distribution of $\Pi(\max_v+1)$ is not independent of $\Pi(i)$ for $i
\equiv \max_v$. The bound $i < n/3$, however, forces the dependence to
be rather small (there are not too many such sets). This greatly helps
in the potential function analysis.

We prove the following lemma with arguments similar to those in the proof of Theorem~\ref{thm:knowledge}.
\begin{lemma}
\label{lem:nested}
	The candidate sets have the following properties:
		
		(i) Two candidate sets $V_i^v$ and $V_j^v$ with $i<j \le \max_v$ and
			$i\not \equiv j$ are disjoint.
	
		(ii) An active candidate set $V_j^v$ is disjoint from any candidate
			set $V_i^v$ provided $i < j < \max_v$.
	
		(iii) The candidate set $V_{i}^v$, $i\le \max_v$ is contained in the
			set $V_{\max_v+1}^{v}$ if $i \equiv \max_v$ and is disjoint
			from it if $i \not \equiv \max_v$.
	
		(iv) For two candidate sets $V_i^v$ and $V_j^v$, $i < j$, if $V_i^v \cap V_j^v \not = \emptyset$ then
            $V_i^v \subset V_j^v$.
\end{lemma}

\begin{proof}
    Let $w$ be the ancestor of $v$ where the function returns score $\max_v$.

    To prove (i), observe that in $w$, one of $V_i^w$ and $V_j^w$ is intersected with $P_0^w$
    while the other is intersected with $P_1^w$ and thus they are made disjoint.  

    To prove (ii), we can assume $i \equiv j$ as otherwise the result follows from
    the previous case. 
    Let $u$ be the ancestor of $v$ such that $F(x_u) = j-1$ and
    that in any ancestor of $u$, the score returned by the function 
    is smaller than $j-1$.  
    At $u$, $V_j^u$ is
    intersected with $P_{j-1 \bmod 2}^u$ while $V_i^v$ is intersected with
    $P_{i \bmod 2}^u$.  Since $ i\equiv j$, it follows that they are again
    disjoint.

    For (iii), the latter part follows as in (i). 
    Consider an ancestor $v'$ of $v$ and let
    $w_j$ be the $j$th child of $v'$ that is also an ancestor of $v$.
    We use induction and we assume $V_{i}^{v'} \subset V_{\max_v+1}^{v'}$.
    If $j<\max_v$, then $V_{\max_v+1}^{v'} = V_{\max_v+1}^{w_j}$ which means
    $V_{i}^{w_j} \subset V_{\max_v+1}^{w_j}$.
    If $j=\max_v$, then $V_{\max_v+1}^{w_j} = V_{\max_v+1}^{v'} \cap P_{\max_v (\mod 2)}^{v'}$ and notice that
    in this case also $V_{i}^{w_j} = V_{i}^{v'} \cap P_{i (\mod 2)}^{v'}$ which still implies
    $V_{i}^{w_j} \subset V_{\max_v+1}^{w_j}$.

    To prove (iv), first observe that the statement is trivial if $i \not \equiv j$.
    Also, if the function returns score $j-1$ at any ancestor of $v$, then by the
    same argument used in (ii) it is possible to show that $V_i^v \cap V_j^v = \emptyset$. 
    Thus assume $i \equiv j$ and the function never returns value $j-1$. In this case,
    it is easy to see that an inductive argument similar to (iii) proves that 
    $V_i^{v'} \subset V_j^{v'}$ for every ancestor $v'$ of $v$.
\end{proof}

\begin{corollary}
\label{lem:disjoint}
Every two distinct active candidate sets $V_i^v$ and $V_j^v$ are disjoint.
\end{corollary}

Remember that the permutation $\Pi$ was chosen uniformly and randomly.
Soon, we shall see that this fact combined with the above properties imply that
$\Pi(i)$ is uniformly distributed in $V_i^v$, when $V_i^v$ is active. 
The following lemma is needed to prove this.

\begin{lemma}\label{lem:distribution}
	Consider a candidate set $V_i^v$ and let $i_1 <  \cdots <  i_k < i$ be the indices of
	candidate sets that are subsets of $V_i^v$.
    Let $\sigma := (\sigma_1, \cdots, \sigma_{i})$ be a sequence without repetition 
	from $[n]$ and let $\sigma' := (\sigma_1, \cdots, \sigma_{i-1})$. 
	Let $n_{\sigma}$ and $n_{\sigma'}$ be the number of permutations in $S_v$ that have $\sigma$ 
	and $\sigma'$ as a prefix, respectively. 
	If $n_\sigma > 0$, then $n_{\sigma'} = (|V_i^v|-k)n_\sigma$.
\end{lemma}

\begin{proof}
	Consider a permutation $\pi \in S_v$ that has $\sigma$ as a prefix.
	This implies $\pi(i) \in V_i^v$.
	For an element $s \in V_i^v$, $s \not = i_j$, $1 \le j \le k$, 
	let $\pi_s$ be the permutation obtained from $\pi$ by placing 
	$s$ at position $i$ and placing $\pi(i)$ where $s$ used to be.
	Since  $s \not = i_j$, $1 \le j \le k$, it follows that $\pi_s$ 
	has $\sigma'$ as prefix and since $s \in V_i^v$ it follows that 
	$\pi_s \in S_v$. It is easy to see that for every permutation in $S_v$
	that has $\sigma$ as a prefix we will create $|V_i^v|-k$ different
	permutations that have $\sigma'$ as a prefix and all these permutations
	will be distinct. 
	Thus, $n_{\sigma'} = (|V_i^v|-k)n_\sigma$.
\end{proof}

\begin{corollary}
\label{cor:uniform}
	Consider a candidate set $V_i^v$ and let $i_1 <  \cdots <  i_k < i$ be the indices of
	candidate sets that are subsets of $V_i^v$.
    Let $\sigma' := (\sigma_1, \cdots, \sigma_{i-1})$ be a sequence without repetition 
	from $[n]$ and let $\sigma_1 := (\sigma_1, \cdots, \sigma_{i-1}, s_1)$ and 
	$\sigma_2 := (\sigma_1, \cdots, \sigma_{i-1}, s_2)$ in which $s_1, s_2 \in V_i^v$.
	Let $n_{\sigma_1}$ and $n_{\sigma_2}$ be the number of permutations in $S_v$ that have $\sigma_1$ 
	and $\sigma_2$ as a prefix, respectively. 
    If $n_{\sigma_1}, n_{\sigma_2} > 0$, then $n_{\sigma_1} = n_{\sigma_2}$.
\end{corollary}

\begin{proof}
	Consider a sequence $s_1,\cdots, s_i$ without repetition from $[n]$ such that
	$s_j \in V_j^v$, $1 \le j \le i$.
	By the previous lemma
	$\Pr[ \Pi(1) = s_1 \wedge \cdots \wedge \Pi(i-1) = s_{i-1} 
	\wedge \Pi(i) = s_i] = \Pr[ \Pi(1) = s_1 \wedge \cdots \wedge \Pi(i-1) = s_{i-1}] \cdot
	(1/|V_i^v|)$. 
\end{proof}

\begin{corollary}
\label{cor:distribution}
	If $V_i^v$ is active, then we have:
	\begin{itemize}
		\item[(i)] $\Pi(i)$ is independent of $\Pi(1), \cdots, \Pi(i-1)$.
		\item[(ii)] $\Pi(i)$ is uniformly distributed in $V_i^v$.
	\end{itemize}
\end{corollary}

\subsection{The Potential Function}
We define the potential of an active candidate set $V_i^v$ as $\log\log{(2n/|V_i^v|)}$.
This is inspired by the upper bound:
a potential increase of $1$ corresponds to a candidate set advancing one level in the 
upper bound context (in the beginning, a set $V_i^v$ has size $n$ and thus its potential is
0 while at the end its potential is $\Theta(\log\log n)$. 
With each level, the quantity $n$ divided by the size of $V_i$ is squared). 
We define the potential at a node $v$ as 
$$\phi(v) = \log\log{\frac{2n}{|V_{\max_v+1}^v|-\con_v}} + \sum_{j \in A_v} \log\log{\frac{2n}{|V_{j}^v|}}\,,$$
in which $A_v$ is the set of indices of active candidate sets at $v$ and
$\con_v$ is the number of candidate sets contained inside $V_{\max_v+1}^v$.
Note that from Lemma~\ref{lem:nested}, it follows that $\con_v = \lfloor \max_v / 2 \rfloor$. 

The intuition for including the term $\con_v$ is the same as our
requirement $i<n/3$ in the definition of active candidate sets, namely
that once $\con_v$ approaches $|V_{\max_v+1}^v|$, the distribution of
$\Pi(\max_v+1)$ starts depending heavily on the candidate sets $V_i^v$
for $i\leq \max_v$ and $i \equiv \max_v$. Thus we have in some sense
determined $\Pi(\max_v+1)$ already when $|V_{\max_v+1}^v|$ approaches
$\con_v$. Therefore, we have to take this into
account in the potential function since otherwise changing $V_{\max_v+1}^v$
from being pseudo-active to being active could give a huge potential
increase.

The following is the main lemma that we wish to prove; it tells
us that the expected increase of the potential function
after each query is constant. The proof of this lemma is the main part of the proof of Theorem~\ref{thm:lower}.


\begin{lemma}\label{lem:pot}
Let $v$ be a node in $T$ and let $\is_v$ be the random variable giving
the value of $F(x_v)$ when $\Pi \in S_v$ and $0$ otherwise. Also
let $w_0,\dots,w_n$ denote the children of $v$, where $w_j$ is the
child reached when $F(x_v)=j$. Then,
$\E[\phi(w_{\is_v})-\phi(v) \mid \Pi \in S_v] = O(1)$.
\end{lemma}

Note that we have $\E[\phi(w_{\is_v})-\phi(v) \mid \Pi \in S_v] = \sum_{a=0}^n \Pr[F(x_v)=a \mid \Pi \in S_v](\phi(w_a)-\phi(v)).$
We consider two main cases: 
$F(x_v) \le \max_v$ and $F(x_v) > \max_v$.
In the first case, the maximum score will not increase in $w_a$ which means $w_a$ will have the
same set of active candidate sets. In the second case,
the pseudo-active candidate set
$V_{\max_v+1}^v$ will turn into an active set $V_{\max_v+1}^{w_a}$ at $w_a$ and $w_a$ will 
have a new pseudo-active set. 
While this second case looks more complicated, it is in fact the
less critical part of the analysis. This is because the probability of suddenly increasing 
the score by a large $\alpha$ is extremely small (we will show that it is roughly $2^{-\Omega(\alpha)}$)
which subsumes any significant potential increase for values of $a > \max_v$.

Let $a_1, \dots, a_{|A_v|}$ be the indices of active candidate sets at $v$ sorted in increasing order.
We also define $a_{|A_v|+1} = \max_v+1$.
For a candidate set $V_{i}^v$, and a Boolean $b \in \left\{ 0,1 \right\}$, let
$V_{i}^v(b) = \left\{ j \in V_{i}^v \mid x_v[j] =b \right\}$.
Clearly, $|V_{i}^v(0)|  +| V_{i}^v(1)|  =| V_{i}^v|$.
For even $a_i$, $1 \le i \le |A_v|$, let 
$$\eps_i = |V_i^v(1)|/|V_i^v|,$$ 
and for odd $i$, let 
$$\eps_i = |V_i^v(0)|/|V_i^v|.$$
Thus $\eps_i$ is the fraction of locations in $V_i^v$ that contain values that does not match $Z_{\Pi(i)}$. Also, we define
$$\eps'_i := \Pr[ a_i \le F(x_v) < a_{i+1}-1 | \Pi \in S_v \wedge F(x_v) \ge a_i].$$
Note that $\eps'_i=0$ if $a_{i+1} = a_i+1$.

With these definitions, it is clear that we have
\begin{eqnarray*}
|V_{a_i}^{w_j}| = (1-\eps_i)|V_{a_i}^v| &\textrm{for}& a_i\leq j.\\
|V_{a_i}^{w_j}| = \eps_i|V_{a_i}^v| &\textrm{for}& a_i=j+1.\\
|V_{a_i}^{w_j}|=|V_{a_i}^v| &\textrm{for}& a_i > j+1.
\end{eqnarray*}
The important fact is that we can also bound other probabilities using the values of the 
$\eps_i$ and $\eps'_i$: We can show that (details follow)
\begin{equation}
\Pr[ F(x_v) = a_{i}-1 | \Pi \in
S_v] \le
\eps_i\Pi_{j=1}^{i-1}(1-\eps_j)(1-\eps'_j) \label{eq:prob}
\end{equation}
and
\begin{eqnarray*}
\Pr[a_{i} \le F(x_v) < a_{i+1}-1 | \Pi \in S_v] \le
\eps'_i(1-\eps_i)\Pi_{j=1}^{i-1}(1-\eps_j)(1-\eps'_j).
\end{eqnarray*}
Thus we have bounds on the changes in the size of the active candidate
sets in terms of the values $\eps_i$. The probability of making the various
changes is also determined from the values $\eps_i$ and $\eps_i'$ and
finally the potential function is defined in terms of the sizes of
these active candidate sets. Thus proving Lemma~\ref{lem:pot} reduces
to proving an inequality showing that any possible choice of the values for the  
$\eps_i$ and $\eps'_i$ provide only little expected increase in
potential.

\paragraph{Proof Sketch.} Since the full calculations are rather lengthy, in the following paragraphs 
we will provide the heart of the analysis by making simplifying assumptions that 
side-step some uninteresting technical difficulties that are needed for the full proof.
We assume that all values $\eps'_i$ are $0$ or in other words, 
$a_i = i$ for all $i \leq \max_v$. Also, we ignore the term in
$\phi(v)$ involving $\con_v$ and consider only the case where the
score returned is no larger than $\max_v$ and $\max_v=n/4$.

Thus the expected increase
in potential provided by the cases where the score does not increase
is bounded by
\begin{eqnarray*}
\sum_{j \leq n/4} \Pr[F(x_v)=j \mid \Pi \in S_v](\phi(w_j)-\phi(v)) \leq
\sum_{j \leq n/4} \eps_j (\phi(w_j)-\phi(v)).
\end{eqnarray*}
Also, we get that when a score of $j$ is returned, we update
\begin{eqnarray*}
|V_i^{w_j}| = (1-\Theta(1/n))|V_i^v| &\textrm{for}& i\leq j.\\
|V_i^{w_j}| = \Theta(|V_i^v|/n) &\textrm{for}& i=j+1.\\
|V_i^{w_j}|=|V_i^v| &\textrm{for}& i > j+1.
\end{eqnarray*}
Examining $\phi(w_j)-\phi(v)$ and the changes to the candidate sets,
we get that there is one candidate set whose size decreases by a
factor $\eps_j$, and there are $j$ sets that change by a factor
$1 - \eps_j$. Here we only consider the potential change caused by the
sets changing by a factor of $\eps_j$. This change is bounded by
\begin{eqnarray*}
\sum_{j \leq n/4} \eps_j \lg \left( \frac{\lg(2n/(\eps_j|V^v_j|)}{\lg(2n/|V^v_j|)}\right) =
\sum_{j \leq n/4} \eps_j \lg \left(1+ \frac{\lg(1/\eps_j)}{\lg(2n/|V^v_j|)}\right) \le 
\sum_{j \leq n/4} \eps_j \lg \frac{\lg(1/\eps_j)}{\lg(2n/|V^v_j|)},
\end{eqnarray*} 
where we used $\log(1 + x) \le x$ for $x > 0$ for the last inequality. The function $\eps_j \rightarrow (1/\eps_j)^{\eps_j}$ is decreasing for $0 < \eps \le 1$. Also, if $\eps > 0$ then $\eps \ge 1/n$ by definition of $\eps$. We can therefore upper bound the sum by setting $\eps_j = 4/n$ for all $j$. To continue these calculations, we use Lemma~\ref{lem:nested} to
conclude that the active candidate sets are disjoint and hence the
sum of their sizes is bounded by $n$. We now divide the sum into
summations over indices where $|V_j^v|$ is in the range $[2^{i} :
  2^{i+1}]$ (there are at most $n/2^i$ such indices):
\begin{eqnarray*}
\sum_{j \leq n/4} \Theta\left(\frac{1}{n}\right) \frac{\lg(n/4)}{\lg(2n/|V^v_j|)}
\leq
\Theta\left(\sum_{i=0}^{\lg n-1}\frac{\lg(n/4)}{2^i\lg(n/2^{i})} \right).\\
\end{eqnarray*}
Now the
sum over the terms where $i>\lg \lg n$ is clearly bounded by a constant
since the $2^i$ in the denominator cancels the $\lg(n/4)$ term and we
get a geometric series of this form. For $i<\lg \lg n$, we have
$\lg(n/4)/\lg(n/2^i) = O(1)$ and we again have a geometric series
summing to $O(1)$.

The full proof is very similar in spirit to the above, just
significantly more involved due to the unknown values $\eps_i$ and
$\eps_i'$. The complete proof is given in the next section.

\subsection{Formal Proof of Lemma~\ref{lem:pot}}
\label{FormalProofLemmaPot}

As we did in the proof sketch we write
\begin{eqnarray*}
\E[\phi(w_{\is_v})-\phi(v) \mid \Pi \in S_v] =
\sum_{a=0}^n \Pr[F(x_v)=a \mid \Pi \in S_v](\phi(w_a)-\phi(v)).
\end{eqnarray*}

As discussed, we divide the above summation into two parts: one for $a \le \max_v$ and another for $a > \max_v$:
\begin{eqnarray}
\notag \E[\phi(w_{\is_v})-\phi(v) \mid \Pi \in S_v] = \\
\sum_{a=0}^{\max_v} \Pr[F(x_v)=a \mid \Pi \in S_v](\phi(w_a)-\phi(v)) + \label{eq:alemax}\\
\sum_{a=1}^{{n/3}-\max_v} \Pr[F(x_v)={\max}_v + a \mid \Pi \in S_v](\phi(w_a)-\phi(v)) \label{eq:agmax}.
\end{eqnarray}

To bound the above two summations, it is clear that we need to handle
$\Pr[F(x_v)=a | \Pi \in S_v]$. In the next section, we will prove lemmas that will
do this.

\paragraph{Bounding Probabilities}
Let $a_1, \dots, a_{|A_v|}$ be the indices of active candidate sets at $v$ sorted in increasing order.
We also define $a_{|A_v|+1} = \max_v+1$.
For a candidate set $V_{i}^v$, and a Boolean $b \in \left\{ 0,1 \right\}$, let
$V_{i}^v(b) = \left\{ j \in V_{i}^v \mid x_v[j] =b \right\}$.
Clearly, $|V_{i}^v(0)|  +| V_{i}^v(1)|  =| V_{i}^v|$.
For even $a_i$, $1 \le i \le |A_v|$, let $\eps_i = |V_i^v(1)|/|V_i^v|$ 
and for odd $i$ let $\eps_i = |V_i^v(0)|/|V_i^v|$.
This definition might seem strange but is inspired by the following observation.
\begin{lemma}\label{lem:cond}
    For $i\le |A_v|$, $\Pr[F(x_v)=a_i-1 | \Pi \in S_v \wedge  F(x_v)> a_i-2] =  \eps_{i}$.
\end{lemma}
\begin{proof}
    Note that $F(x_v)=a_i-1$ happens if and only if $F(x_v)
      > a_i - 2$  and
    $x_v[\Pi(a_i)]\not \equiv a_i$.
    Since $V_{a_i}^v$ is an active candidate set, 
    the lemma follows from Corollary~\ref{cor:distribution} and the definition of
    $\eps_{i}$.
\end{proof}

Let $\eps'_i := \Pr[ a_i \le F(x_v) < a_{i+1}-1 | \Pi \in S_v \wedge F(x_v) \ge a_i]$.
Note that $\eps'_i=0$ if $a_{i+1} = a_i+1$.

\begin{lemma}
\label{lem:probLEmax}
    For $i \le |A_v|$ we have 
$|V_{a_i}^{w_j}| = |V_{a_i}^{v}|$ for $j < a_i-1$,
    $|V_{a_i}^{w_j}| = \eps_i |V_{a_i}^{v}|$ for  $j = a_i-1$, and
    $|V_{a_i}^{w_j}| = (1-\eps_i)|V_{a_i}^{v}|$ for $j > a_i-1$. Also,
\begin{align}
    \Pr[ F(x_v) = a_{i}-1 | \Pi \in S_v] &= \eps_i\Pi_{j=1}^{i-1}(1-\eps_j)(1-\eps'_j),  \label{eq:eq1} \\
    \Pr[a_{i} \le F(x_v) < a_{i+1}-1 | \Pi \in S_v] &= 
\eps'_i(1-\eps_i)\Pi_{j=1}^{i-1}(1-\eps_j)(1-\eps'_j), \label{eq:eq2}
\end{align} 
\end{lemma}

\begin{proof}
    Using Lemma~\ref{lem:cond}, it is verified that
\begin{align*}
\Pr[ F(x_v) > a_{i}-1 | \Pi \in S_v]
&=
\Pr[ F(x_v) > a_{i}-2 \wedge F(x_v) \neq a_{i}-1 | \Pi \in S_v]\\ 
&=
\Pr[F(x_v) \not = a_i -1|F(x_v) > a_i-2 \wedge \Pi \in S_v] \cdot \\
\Pr[F(x_v)> a_i-2| \Pi \in S_v] &= 
(1-\eps_i) \Pr[F(x_v)> a_i-2| \Pi \in S_v].
\end{align*}
Similarly, using the definition of $\eps'_i$ we can see that
\begin{align*}
& \Pr[ F(x_v) > a_{i}-2| \Pi \in S_v] \\
& =
\Pr[F(x_v) \not \in \left\{ a_{i-1}, \dots, a_i-2 \right\} \wedge
F(x_v) > a_{i-1}-1 | \Pi \in S_v] \\
&=
\Pr[F(x_v) \not \in \left\{ a_{i-1}, \dots, a_i-2 \right\} | F(x_v) > a_{i-1}-1 \wedge 
\Pi \in S_v] \Pr[F(x_v) > a_{i-1}-1 | \Pi \in S_v] \\
&=
(1-\eps'_{i-1}) \Pr[F(x_v) > a_{i-1}-1 | \Pi \in S_v].
\end{align*}
Using these equalities, we find that
\begin{equation*}
    \Pr[ F(x_v) > a_{i}-1 | \Pi \in S_v] = (1-\eps_i)\Pi_{j=1}^{i-1}(1-\eps_j)(1-\eps'_j)
    \label{eq:gi}
\end{equation*}
and
\begin{equation*}
    \Pr[ F(x_v) > a_{i}-2 | \Pi \in S_v] = \Pi_{j=1}^{i-1}(1-\eps_j)(1-\eps'_j).
    \label{eq:gii}
\end{equation*}
Equality (\ref{eq:eq1}) follows from the above two equalities. 
To obtain Equality (\ref{eq:eq2}), we only need to consider the definition of $\eps'_i$ as an additional ingredient.
The rest of the lemma follows directly from the definition of $\eps_i$ and the candidate sets.
\end{proof}

\begin{lemma}
\label{lem:aeqmax}
    Let $b \in \left\{ 0,1 \right\}$ be such that $b \equiv \max_v$ and let $k := |V_{\max_v+1}^v(b)|$.
    Then,
\begin{eqnarray*}
    \Pr[F(x_v) = {\max}_v | \Pi \in S_v] =
{ k-\con_v \over |V_{\max_v+1}^v|-\con_v } \prod_{i=1}^{|A_v|}(1-\eps_i)(1-\eps'_i).
\end{eqnarray*}
\end{lemma}

\begin{proof}
Conditioned on $F(x_v) > \max_v-1$, $F(x_v)$ will be equal to $\max_v$ if 
    $x_v[\Pi(\max_v+1)] = b$.
    By definition, the number of positions in $V_{\max_v+1}^v$ that satisfy this is $k$.
    However, $V_{\max_v+1}^v$ contains $\con_v $ candidate sets but since
    $V_{\max_v+1}^v$ can only contain a candidate set $V_{i}^v$ if $i \equiv \max_v$ (by Lemma~\ref{lem:nested}),
    it follows from Lemma~\ref{lem:distribution} that 
    $\Pr[F(x_v) = \max_v | \Pi \in S_v \wedge F(x_v) > \max_v-1] = (k-\con_v )/(|V_{\max_v+1}^v|-\con_v )$.
    The lemma then follows from the previous lemma.
\end{proof}

\begin{lemma}
\label{lem:aggmax}
    Let $b \in \left\{ 0,1 \right\}$ be such that $b \equiv \max_v$ and let $k := |V_{\max_v+1}^v(b)|$.
    Then,
\begin{eqnarray*}
    \Pr[F(x_v) > {\max}_v | \Pi \in S_v] \le {|V_{\max_v+1}^v|- k \over |V_{\max_v+1}^v|-\con_v }\prod_{i=1}^{|A_v|}(1-\eps_i)(1-\eps'_i).
\end{eqnarray*}
\end{lemma}

\begin{proof}
    From the previous lemma we have that 
    $\Pr[F(x_v) = \max_v | \Pi \in S_v \wedge F(x_v) > \max_v-1] = (k-\con_v )/(|V_{\max_v+1}^v|-\con_v )$.
    Thus,
\begin{eqnarray*}
\Pr[F(x_v) > {\max}_v | \Pi \in S_v \wedge F(x_v) > {\max}_v-1] =
{|V_{\max_v+1}^v|- k  \over |V_{\max_v+1}^v|-\con_v }.
\end{eqnarray*}
The claim now follows from Lemma~\ref{lem:probLEmax}. 
\end{proof}

Remember that $P_0^v$ (resp. $P_1^v$) are the set of positions in $x_v$ that contain
0 (resp. 1).

\begin{lemma}\label{lem:agmax}
    Let $x_0 = |P_0^v|$ and $x_1 = |P_1^v|$.
    Let $b_i$ be a Boolean such that $b_i \equiv i$.
    For $a \geq \max_v+2$
    \begin{eqnarray*}
        \Pr[ F(x_v) = a \mid \Pi \in S_v ] \le
  \left(   \prod_{i=\max_v + 2}^{a}{x_{b_i} - \lfloor i/2 \rfloor \over n-i+1}\right)\left(1-{x_{b_{a+1}} - \lfloor (a+1)/2 \rfloor \over n-a} \right).
    \end{eqnarray*}
\end{lemma}

\begin{proof}
    Notice that we have $V_{i}^v = [n]$ for $i\ge\max_v+2$ which means $i-1$ is the number of candidates sets
    $V_{j}^{v}$ contained in $V_{i}^v$ and among those $\lfloor i/2 \rfloor$ are such that $i\equiv j$.
    Consider a particular prefix $\sigma=(\sigma_1, \cdots, \sigma_{i-1})$ such that there exists a permutation 
    $\pi \in S_v$ that has $\sigma$ as a prefix.
    This implies that $\sigma_j \in P_{b_j}^v$.
    Thus, it follows that there are $x_{b_i} - \lfloor i/2 \rfloor$ elements $s \in P_{b_i}^v$ such that
    the sequences $(\sigma_1, \cdots, \sigma_{i-1}, s)$ can be the prefix
    of a permutation in $S_v$.
    Thus by Corollary~\ref{cor:uniform}, and for $i \ge \max_v+2$, 
    \[ \Pr[ F(x_v) = i-1  | \Pi \in S_v \wedge F(x_v) \ge i-1 ] = 1-{x_{b_i} - \lfloor i/2 \rfloor \over n-i+1} \]
    and
    \[ \Pr[ F(x_v) \ge i  | \Pi \in S_v \wedge F(x_v) \ge i-1 ] = {x_{b_i} - \lfloor i/2 \rfloor \over n-i+1}. \]
\end{proof}

\begin{corollary}\label{cor:agmax}
	For $\max_v+1 \le a \le n/3$ we have 
\begin{eqnarray*}
\Pr[ F(x_v) = a \mid \Pi \in S_v ] =
2^{-\Omega(a-\max_v)} \cdot\left(1-{x_{b_{a+1}} - \lfloor (a+1)/2 \rfloor \over n-a} \right).
\end{eqnarray*}
\end{corollary}

\begin{proof}
    Since $x_{b_i} + x_{b_{i+1}} = n$, it follows that
\begin{eqnarray*}
\left( {x_{b_i} - \lfloor i/2 \rfloor \over n-i+1} \right) \left( {x_{b_{i+1}} - \lfloor i/2 \rfloor \over n-i+2} \right) \le
\left( {x_{b_i} - \lfloor i/2 \rfloor \over n-i+1} \right) \left( {x_{b_{i+1}} - \lfloor i/2 \rfloor \over n-i+1} \right) \le {1 \over 2}.
\end{eqnarray*}
\end{proof}

Now we analyze the potential function. 
Remember that we have split the change in the potential function into two
parts: one is captured by Equation~\eqref{eq:alemax} and the other
by Equation~\eqref{eq:agmax}.
We claim each is bounded by a constant. 

\begin{lemma}\label{lem:alemax}
  The value of Equation~\eqref{eq:alemax} is bounded by $O(1)$.
\end{lemma}
\begin{proof}

We have,
\begin{equation}
\phi(w_a) - \phi(v) =
	  \log {\log{2n \over |V_{\max_{w_a}+1}^{w_a}|-\con_{w_a}} \over \log{2n \over |V_{\max_v+1}^v|-\con_v }} + \sum_{j\in A_v} \log{\log{2n \over |V_{j}^{w_a}|} \over  \log{2n \over |V_{j}^v|}}.
    \label{eq:ipot}
\end{equation}
When $a \le \max_v$ we have, $\max_v = \max_{w_a}$ and
$\con_v = \con_{w_a}$. 
For $a < \max_v$, we also have $V^{w_a}_{\max_v + 1} = V^{v}_{\max_v + 1}$.
It is clear from (\ref{eq:ipot}) that for $a_i \le a < a_{i+1}-1$, all the values
of $\phi(w_a) -\phi(v)$  will be equal. Thus, (\ref{eq:alemax}) is equal to

\begin{eqnarray}
 &&  \sum_{i=1}^{|A_v|} \Pr[F(x_v) = a_i-1 | \Pi \in S_v] (\phi(w_{a_i-1})-\phi(v)) + \label{eq:case1-1} \\
    & & \sum_{i=1}^{|A_v|} \Pr[ a_i \le F(x_v) < a_{i+1}-1 | \Pi \in S_v] (\phi(w_{a_i})-\phi(v)) +\label{eq:case1-2} \\
    & & \Pr[F(x_v)={\max}_v  | \Pi \in S_v] (\phi(w_{\max_v}) - \phi(v))\label{eq:case1-3}.
\end{eqnarray}

We now analyze each of these terms separately.

\paragraph{Analyzing (\ref{eq:case1-1})}
We write (\ref{eq:case1-1}) using (\ref{eq:ipot}) and Lemma~\ref{lem:probLEmax}(\ref{eq:eq1}). 
Using inequalities, $1-x \le e^{-x}$, for $0 \le x \le 1$,
$\log(1+x) \le x$ for $x\ge 0$, and $\sum_{1 \le i \le k} y_i \log 1/y_i
  \le Y \log (k/Y)$ for $y_i \ge 0$ and $Y = \sum_{1 \le i \le k} y_i$, we obtain that (\ref{eq:case1-1}) equals 
\begin{align}
    \notag  
    & \mathlarger{\sum}_{i=1}^{|A_v|}   \eps_i\prod_{j=1}^{i-1}(1-\eps_j)(1-\eps'_j)\left( \sum_{j=1}^{i} \log{\log{2n \over |V_{a_j}^{w_{a_i-1}}|} \over  \log{2n \over |V_{a_j}^v|}}\right) 
    \\  =     
    \notag &\mathlarger{\sum}_{i=1}^{|A_v|}   \eps_i\prod_{j=1}^{i-1}(1-\eps_j)(1-\eps'_j)\left(\log{\log{2n \over |V_{a_i}^{w_{a_i-1}}|} \over  \log{2n \over |V_{a_i}^v|}}  +\sum_{j=1}^{i-1} \log{\log{2n \over |V_{a_j}^{w_{a_i-1}}|} \over  \log{2n \over |V_{a_j}^v|}}\right) 
    \\  = 
    \notag &\mathlarger{\sum}_{i=1}^{|A_v|}   \eps_i\prod_{j=1}^{i-1}(1-\eps_j)(1-\eps'_j)\left(\log{\log{2n \over |V_{a_i}^{v}|}+ \log{1\over \eps_i}\over  \log{2n \over |V_{a_i}^v|}}  +\sum_{j=1}^{i-1} \log{\log{2n \over |V_{a_j}^{v}|} + \log{1\over 1-\eps_j}\over  \log{2n \over |V_{a_j}^v|}}\right)  
    \\  = 
    \notag &\mathlarger{\sum}_{i=1}^{|A_v|}   \eps_i\prod_{j=1}^{i-1}(1-\eps_j)(1-\eps'_j)\left(\log\left(  1+{\log{1\over \eps_i}\over  \log{2n \over |V_{a_i}^v|}} \right)  +\sum_{j=1}^{i-1} \log\left(  1+{\log{1\over 1-\eps_j}\over  \log{2n \over |V_{a_j}^v|}}\right) \right)   
    \\  \le 
    \notag &\mathlarger{\sum}_{i=1}^{|A_v|}   \eps_i\prod_{j=1}^{i-1}(1-\eps_j)\left(\log\left(  1+{\log{1\over \eps_i}\over  \log{2n \over |V_{a_i}^v|}} \right)  +\sum_{j=1}^{i-1} \log\left(  1+{\log{1\over 1-\eps_j}}\right) \right)   
    \\ \le 
    \notag &\mathlarger{\sum}_{i=1}^{|A_v|}   \eps_i\prod_{j=1}^{i-1}(1-\eps_j)\left(\log\left(  1+{\log{1\over \eps_i}\over  \log{2n \over |V_{a_i}^v|}} \right)  +\sum_{j=1}^{i-1} {\log{1\over 1-\eps_j}} \right)    
    \\  = 
    &\mathlarger{\sum}_{i=1}^{|A_v|}   \eps_i\log\left(  1+{\log{1\over \eps_i}\over  \log{2n \over |V_{a_i}^v|}} \right)\prod_{j=1}^{i-1}(1-\eps_j)   + \label{eq:case1-1-p1} \\
 & \quad     \mathlarger{\sum}_{i=1}^{|A_v|}   \eps_i\prod_{j=1}^{i-1}(1-\eps_j)\left( \sum_{j=1}^{i-1} {\log{1\over 1-\eps_j}} \right). \label{eq:case1-1-p2}
\end{align}
To bound (\ref{eq:case1-1-p1}), we use the fact that any two active candidate sets
are disjoint. We break the summation into
smaller chunks. Observe that 
$\prod_{j=1}^{i-1}(1-\eps_j) \le  e^{-\sum_{j=1}^{i-1}\eps_j}$.
Thus, let $J_t$, $t \ge 0$, be the set of indices such that for each $i\in J_t$ we have
$2^t - 1 \le \sum_{j=1}^{i-1}\eps_j  < 2^{t+1}$.
Now define $J_{t,k} = \left\{i \in J_t \mid  n/2^{k+1} \le |V_{a_{i}}^v| \le n/2^k   \right\}$, for $ 0 \le k \le \log n$
and let $s_{t,k} = \sum_{i \in J_{t,k}} \eps_i$.
Observe that by the disjointness of two active candidate sets, $|J_{t,k}| \le
2^{k+1}$. The sum (\ref{eq:case1-1-p1}) is thus equal to 
\begin{eqnarray*}
	\sum_{t=0}^{\log n}  \sum_{k=1}^{\log n} \sum_{i \in J_{t,k}}   \eps_i\log\left(  1+{\log{1\over \eps_i}\over  \log{2n \over |V_{a_i}^v|}} \right)\prod_{j=1}^{i-1}(1-\eps_j) \le  \\
	\sum_{t=0}^{n}  \sum_{k=1}^{\log n} \sum_{i \in J_{t,k}}   \eps_i\log\left(  1+{\log{1\over \eps_i}\over  k} \right)e^{-\sum_{j=1}^{i-1}\eps_j} \le  \\
	\sum_{t=0}^{n}  \sum_{k=1}^{\log n} \sum_{i \in J_{t,k}}   \eps_i{\log{1\over \eps_i}\over  k}e^{-2^t+1} \le \sum_{t=0}^{n}  \sum_{k=1}^{\log n} {s_{t,k} \log{|J_{t,k}| \over s_{t,k}} \over k} e^{-2^t+1} \le  \\
	\sum_{t=0}^{n}  \sum_{k=1}^{\log n} {s_{t,k} (k+1)+s_{t,k}\log{1 \over s_{t,k}} \over k} e^{-2^t+1}
	\le\\  \sum_{t=0}^{n}  2^{t+2} e^{-2^t+1} + \sum_{t=0}^n\sum_{k=1}^{\log n} {{s_{t,k} \log{1 \over s_{t,k}} \over k}} e^{-2^t+1} \le  \\
	O(1) + \sum_{t=0}^n\sum_{r=1}^{\lg \lg n}\sum_{k=2^{r-1}}^{2^r}{{s_{t,k} \log{1 \over s_{t,k}} \over 2^{r-1}}} e^{-2^t+1}\,. \\
\end{eqnarray*}  

Now define  $S_{t,r} = \sum_{2^{r -1} \le k < 2^r} s_{t,k}$.
Remember that  $\sum_{r=1}^{\log\log n}S_{t,r} < 2^{t+1}$.
(\ref{eq:case1-1-p1}) is thus at most 

\begin{eqnarray*}
	O(1) + \sum_{t=0}^n\sum_{r=1}^{\lg \lg n}{ {S_{t,r} \log{2^{r-1} \over
              S_{t,r}} \over 2^{r-1}}} e^{-2^t+1}= \\   O(1) + \sum_{t=0}^n\sum_{r=1}^{\lg \lg n}{ {S_{t,r}(r-1) + S_{t,r} \log{1 \over S_{t,r}} \over 2^{r-1}}} e^{-2^t+1} \le \\
	 O(1)+ \sum_{t=0}^n\sum_{r=1}^{\lg \lg n}{ {2^{t+1}(r-1) \over 2^{r-1}}} e^{-2^t+1} + \sum_{t=0}^n\sum_{r=1}^{\lg \lg n}{ {1 \over 2^{r-1}}} e^{-2^t+1}  = O(1). \\
\end{eqnarray*}

To bound (\ref{eq:case1-1-p2}), define $J_t$ as before and let $p_i = \prod_{j=1}^{i-1}1/(1-\eps_j)$.
Observe that the function $\log(1/(1-x))$ is a convex function which means if $s_i := \sum_{j=1}^{i-1}\eps_j$
is fixed, then  $\prod_{j=1}^{i-1}1/(1-\eps_j)$ is minimized when $\eps_j = s_i/(i-1)$. Thus,
$$p_i \ge \left( {1 \over 1- {s_i \over i-1}} \right)^{i-1} \ge  \left( {1+ {s_i \over i-1}} \right)^{i-1} \ge 
1+ {i-1 \choose j} \left( {s_i \over i-1} \right)^j$$
in which $j$ can be chosen to be any integer between 0 and $i-1$.
We pick $j:= \max\left\{ \lfloor s_i / 8\rfloor, 1 \right\}$. 
Since $i \ge s_i$, we get for $i \in J_t$,
\[p_i \ge 1+ { \left( { {i-1} \over 2} \right)^j \over j^j } \left( {s_i \over i-1} \right)^j \ge  \
1+ {\left( s_i \over 2j \right)^j} \ge 2^{s_i/8} \ge 2^{2^{t-4}}.\]
Thus, we can write (\ref{eq:case1-1-p2}) as 
\begin{eqnarray*}
      \sum_{i=1}^{|A_v|}   \eps_i  {\log p_i \over p_i} =  \sum_{t=0}^{\log n}\sum_{i \in J_t}   \eps_i  {\log p_i \over p_i} 
=  \sum_{t=0}^{\log n}\sum_{i \in J_t}  O\left({ \eps_i 2^t  \over 2^{2^{t-4}} } \right) \le 
    \sum_{t=0}^{\log n}  O\left({  2^{2t+1}  \over 2^{2^{t-4}}  }\right) = O(1).
\end{eqnarray*}

\paragraph{Analyzing (\eqref{eq:case1-2})}
The analysis of this equation is very similar to (\ref{eq:case1-1}).
We can write (\ref{eq:case1-2}) using (\ref{eq:ipot}) and (\ref{eq:eq2}). 
We also use the same technique as in analyzing (\ref{eq:case1-1-p2}).
We find that the sum (\ref{eq:case1-2}) equals 
\begin{eqnarray*}
    \sum_{i=1}^{|A_v|} \eps'_i(1-\eps_i)\prod_{j=1}^{i-1}(1-\eps_j)(1-\eps'_j)\left( \sum_{j=1}^{i} \log{\log{2n \over |V_{a_j}^{w_{a_i}}|} \over  \log{2n \over |V_{a_j}^v|}}\right) = \\
    \mathlarger{\sum}_{i=1}^{|A_v|}   \eps_i'(1-\eps_i)\prod_{j=1}^{i-1}(1-\eps_j)(1-\eps'_j)\left(\sum_{j=1}^{i} \log{\log{2n \over |V_{a_j}^{v}|} + \log{1\over 1-\eps_j}\over  \log{2n \over |V_{a_j}^v|}}\right)   = \\
    \mathlarger{\sum}_{i=1}^{|A_v|}   \eps_i'(1-\eps_i)\prod_{j=1}^{i-1}(1-\eps_j)(1-\eps'_j)\left(\sum_{j=1}^{i} \log \left(1+ { \log{1\over 1-\eps_j}\over  \log{2n \over |V_{a_j}^v|}} \right)\right)   \le \\
    \mathlarger{\sum}_{i=1}^{|A_v|}   \eps_i'(1-\eps_i)\prod_{j=1}^{i-1}(1-\eps_j)(1-\eps'_j)\left(\sum_{j=1}^{i} \log{1 \over 1-\eps_i} \right).
\end{eqnarray*}

Let $s_i := \sum_{j=1}^i \eps_j$, and $s'_i := \sum_{j=1}^{i-1} \eps'_j$. 
Similarly to the previous case, let $J_t$, $t \ge 0$, be the set of indices such that for each $i\in J_t$ we have
$2^t - 1 \le s_i + s'_i  \le 2^{t+1}$.
Also define  $p_i = \prod_{j=1}^{i}1/(1-\eps_j)$ and  $p'_i = \prod_{j=1}^{i-1}1/(1-\eps'_j)$.
Using the previous techniques we can get that $p_i \ge 2^{s_i/8}$ and $p'_i \ge 2^{s'_i/8}$.
We get that (\ref{eq:case1-2}) is at most 
\begin{align*}
    &\sum_{t=0}^{\log n}\sum_{i \in J_t} {\eps'_i} {\log p_i \over p_i p'_i}
    \le \sum_{t=0}^{\log n}\sum_{i \in J_t} {\eps'_i} {s_i \over  8\cdot 2^{s_i/8}2^{s'_i/8}} \\
    \le &\sum_{t=0}^{\log n}\sum_{i \in J_t} {\eps'_i} {2^{t+1} \over  8\cdot 2^{(2^{t}-1)/8}} \le 
     \sum_{t=0}^{\log n}{2^{2t+2} \over  8\cdot 2^{(2^{t}-1)/8}} = O(1).
\end{align*}

\paragraph{Analyzing (\ref{eq:case1-3})}
    Let $b \in \left\{ 0,1 \right\}$ be such that $b \equiv \max_v$ and let $k := |V_{\max_v+1}^v(b)|$.
    We use Lemma~\ref{lem:aeqmax}.
 Observe that we still have $\con_{v} = \con_{w_{\max_v}}$.
Thus we can bound (\ref{eq:case1-3}) from above by 
\begin{eqnarray*}
 {(k-\con_v )\over |V_{\max_v+1}^v|-\con_v } \prod_{i=1}^{|A_v|}(1-\eps_i)(1-\eps'_i) (\phi(w_{\max_v}) - \phi(v)) =  \\
{(k-\con_v )\over |V_{\max_v+1}^v|-\con_v }    \prod_{i=1}^{|A_v|}(1-\eps_i)(1-\eps'_i)\left( \log {\log{2n \over |V_{\max_v+1}^{w_{\max_v}}|-\con_{w_{\max_v}}} \over \log{2n \over |V_{\max_v+1}^v|-\con_v }} + \sum_{i=1}^{|A_v|} \log{\log{2n \over |V_{i}^{w_{\max_v}}|} \over  \log{2n \over |V_{i}^v|}} \right) \le \\
{(k-\con_v )\over |V_{\max_v+1}^v|-\con_v }  \prod_{i=1}^{|A_v|}(1-\eps_i)(1-\eps'_i)\left( \log {\log{2n \over k-\con_v } \over \log{2n \over |V_{\max_v+1}^v|-\con_v }} + \sum_{i=1}^{|A_v|} \log \left(1+ \log{1\over 1-\eps_i} \right) \right)\le \\
{(k-\con_v )\over |V_{\max_v+1}^v|-\con_v }  \log {\log{2n \over k-\con_v } \over \log{2n \over |V_{\max_v+1}^v|-\con_v }} + O(1) = O(1).
\end{eqnarray*}

\end{proof}

We now return to the other remaining part of the change in the potential function. 
\begin{lemma}\label{lem:agmax2}
  The value of Equation~\eqref{eq:agmax} is bounded by $O(1)$.
\end{lemma}
\begin{proof}
The big difference here is that the candidate set $V_{\max_v+1}^{w_a}$ becomes an active candidate set (if of course
$\max_v+1 < n/3$) at $w_a$ while $V_{\max_v+1}^{v}$ was not active at $v$.
Because of this, we have 
\begin{align*}
 \phi(w_a) - \phi(v) \le  
 \log \log{2n \over |V_{a+1}^{w_a}|-\con_{w_a}} 
 &+ 
\log\log{2n \over |V_{\max_v+1}^{w_a}|} - 
 \log \log{2n \over |V_{\max_v+1}^{v}|-\con_{v}}\\ 
 & + 
\sum_{j\in A_v} \log{\log{2n \over |V_{j}^{w_a}|} \over  \log{2n \over |V_{j}^v|}}. 
\end{align*}
Thus, we get that
\begin{eqnarray}
    \notag &&\sum_{a>\max_v}^{n/3} \Pr[F(x_v)=a \mid \Pi \in S_v](\phi(w_a)-\phi(v)) = \\
    && \quad \sum_{a>\max_v}^{n/3} \Pr[F(x_v)=a \mid \Pi \in S_v]\log \log{2n \over |V_{a+1}^{w_a}|-\con_{w_a}} + \label{eq:case2-1} \\
    &&\quad \sum_{a>\max_v}^{n/3} \Pr[F(x_v)=a \mid \Pi \in S_v]\log\left( {\log{2n \over |V_{\max_v+1}^{w_a}|}\over \log{2n \over |V_{\max_v+1}^v| -\con_v }} \right) +  \label{eq:case2-2} \\
   &&\quad  \sum_{a>\max_v}^{n/3} \Pr[F(x_v)=a \mid \Pi \in S_v]\sum_{j\in A_v} \log{\log{2n \over |V_{j}^{w_a}|} \over  \log{2n \over |V_{j}^v|} }\label{eq:case2-3}.
\end{eqnarray}

Using the previous ideas, it is easy to see that we can bound (\ref{eq:case2-3}) from above by 
\begin{align*}
    & \sum_{a>{\max}_v}^{n/3} 
    \big(
    \Pr[F(x_v)=a \mid \Pi \in S_v \wedge F(x_v) > {\max}_v]
    \cdot\\
    & \quad \quad \quad \Pr[F(x_v) > {\max}_v \mid \Pi \in S_v]\cdot \sum_{j\in A_v} \log{\log{2n \over |V_{j}^{w_a}|} \over  \log{2n \over |V_{j}^v|}}\big)
    \\ \le 
    & \sum_{a>{\max}_v}^{n/3} 
    \big( \Pr[F(x_v)=a \mid \Pi \in S_v \wedge F(x_v) > {\max}_v]\cdot
    \\ & \quad \quad \quad 
    \prod_{i=1}^{|A_v|}(1-\eps_i)(1-\eps'_i)\sum_{j\in A_v} \log{\log{2n \over |V_{j}^{w_a}|} \over  \log{2n \over |V_{j}^v|}} \big)  
    \\ \le
	 & \prod_{i=1}^{|A_v|}(1-\eps_i)(1-\eps'_i)\sum_{j\in A_v} \log{\log{2n \over |V_{j}^{w_a}|} \over  \log{2n \over |V_{j}^v|}}  
	 \\  \le
	 & \prod_{i=1}^{|A_v|}(1-\eps_i)(1-\eps'_i)\sum_{j\in A_v} \log\left( {1+\log {1 \over 1-\eps_j}} \right) 
	 \\  \le
	 & \prod_{i=1}^{|A_v|}(1-\eps_i)\sum_{j\in A_v} \log {1 \over 1-\eps_j}  
	 \\  \le
	 & \prod_{i=1}^{|A_v|}(1-\eps_i)\log\left( {1 \over \prod_{j\in A_v} (1-\eps_j) } \right) = O(1). \\
\end{align*}

To analyze (\ref{eq:case2-2}) by Lemma~\ref{lem:aggmax} we know that
$\Pr[F(x_v) > \max_v | \Pi \in S_v] \le{|V_{\max_v+1}^v|- k  \over |V_{\max_v+1}^v|-\con_v }$
in which $k$ is as defined in the lemma.
Note that in this case $|V_{\max_v+1}^{w_a}| =|V_{\max_v+1}^v|-k$.
This implies that (\ref{eq:case2-2}) is at most 

\begin{eqnarray*}
     \sum_{a>{\max}_v}^{n} \Pr[F(x_v)=a \mid \Pi \in S_v]\log\left( {\log{2n \over |V_{\max_v+1}^v|-k}\over \log{2n \over |V_{\max_v+1}^v| -\con_v }} \right)  = \\
     \left(  \sum_{a>{\max}_v}^{n} \Pr[F(x_v)=a \mid \Pi \in S_v]\right)\log\left( {\log{2n \over |V_{\max_v+1}^v|-k}\over \log{2n \over |V_{\max_v+1}^v| -\con_v }} \right)  = \\
     \Pr[F(x_v) > {\max}_v \mid \Pi \in S_v]\log\left( {\log{2n \over |V_{\max_v+1}^v|-k}\over \log{2n \over |V_{\max_v+1}^v| -\con_v }} \right)  \le \\
     {|V_{\max_v+1}^v|- k  \over |V_{\max_v+1}^v|-\con_v } \log\left( {\log{2n \over |V_{\max_v+1}^v|-k}\over \log{2n \over |V_{\max_v+1}^v| -\con_v }} \right)  = O(1). \\
\end{eqnarray*}

It is left to analyze (\ref{eq:case2-1}). 
Let $x_0 = |P_0^v|$ and $x_1 = |P_1^v|$.
Let $b_i$ be a Boolean such $b_i \equiv i$.
Note that we have $|V_{\max_v+a+1}^{w_{\max_v+a}}|= x_{b_{{\max_v+a}}} = n-x_{b_{{\max_v+a+1}}}$. 
Using Corollary~\ref{cor:agmax} we the term (\ref{eq:case2-1}) is at most  
\begin{align*}
    &  \sum_{a=1}^{n/3-\max_v} \Pr[F(x_v)={\max}_v+a \mid \Pi \in S_v]\log \log{2n \over |V_{\max_v+a+1}^{w_{\max_v+a}}|-\con_{w_{\max_v+a}}}\\ 
   \le  & 
    \sum_{a=1}^{n/3-\max_v} 2^{-\Omega(a)}\Big(1-{x_{b_{ {\max}_v+a+1}} - \lfloor ( {\max}_v+a+1)/2 \rfloor \over n- {\max}_v-a} \Big)\\
    & \quad \quad \quad \quad  \cdot \log \log{2n \over n-x_{b_{{\max_v+a+1}}} -\lfloor ( {\max}_v+a)/2 \rfloor}\\
   = & O(1). \\
\end{align*}

\end{proof}

As discussed earlier, to prove our main lemma,  Lemma~\ref{lem:pot}, we need to analyze the change in the potential
function when going from node $v$ to a child of $v$.
This change in potential was split into two terms, one captured by Equation~\eqref{eq:alemax} and another
by Equation~\eqref{eq:agmax}.
By Lemmas~\ref{lem:alemax} and~\ref{lem:agmax2}, each of these terms is bounded by $O(1)$ and thus as a result, the change
in the potential when going from $v$ to a child of $v$ is also bounded by $O(1)$. 

To prove our lower bound, we simply need to observe that at the beginning, the potential is 0 since the algorithm
essentially knows nothing. 
However, at the end, the sizes of all the candidates sets are small. 
Consider a node $v$ the decision where we have hit $\max_v = n/3$. 
By Lemma~\ref{lem:nested}, it follows that $\Omega(n)$ candidate sets have $O(1)$ size.
As a result, the value of the potential function at $v$ is $\Omega(n\log \log n)$ as each candidate set of size $O(1)$
contributes $\Omega(\log\log n)$ to the potential. 
However, as our main lemma, Lemma~\ref{lem:pot}, shows, the expected increase in the potential function is $O(1)$.
As a result, it follows that the expected depth of $v$ must be $\Omega(n\log\log n)$.

\subsection{Concluding the Proof of Theorem~\ref{thm:lower}}


Intuitively, if the maximum score value increases after a query, it increases, in expectation, only by an additive constant.
In fact, 
as shown in Corollary~\ref{cor:agmax}, the probability of increasing
the maximum score value by $\alpha$ after one query is $2^{-\Omega(\alpha)}$.
Thus, it follows from the definition of the active candidate sets that
when the score reaches $n/3$ we expect $\Omega(n)$ active candidate sets.
However, by Lemma~\ref{lem:nested}, the active candidate sets are disjoint.
This means that a fraction of them (again at least $\Omega(n)$ of them), must
be small, or equivalently, their total potential is $\Omega(n\log\log n)$, meaning,
at least $\Omega(n \log\log n)$ queries have been asked.
In the rest of this section, we prove this intuition.

Given an input $(z,\pi)$, we say an edge $e$ in the decision tree $T$ is {\emph{increasing}}
if $e$ corresponds to an increase in the maximum score and it is traversed given the input $(z,\pi)$.
We say that an increasing edge is {\emph{short}} if it corresponds to an increase of at most $c$ in the
maximum function score (in which $c$ is a sufficiently large constant) and we call it {\emph{long}} otherwise.
Let $N$ be the random variable denoting the number of
increasing edges seen on input $\Pi$ before reaching a node with score greater than $n/3$.
Let $L_j$ be the random indicator variable taking the value $0$ if the
$j$th increasing edge is short, and taking the value equal to the amount of increase in the score along this edge if not.
If $j>N$, then we define $L_j=0$. 
Also let $W_j$ be the random variable corresponding to the node of the decision tree where the
$j$th increase happens.
As discussed, we have shown that for every node $v$, 
$\Pr[L_j \ge \alpha | W_j = v ] \le 2^{-\Omega(\alpha)}$.
We want to
upper bound $\sum_{j=1}^n \E[L_j]$ (there are always at most $n$
increasing edges). From the above, we know that
\begin{align*}
\E[L_j] 
& \leq \E[L_j \mid N \geq j]
= \sum_{v \in T} \sum_{i=c+1}^n i\cdot \Pr[ L_j = i\wedge W_j = v \mid N \geq j] \\
&= \sum_{v \in T} \sum_{i=c+1}^n i\cdot \Pr[ L_j = i\wedge W_j = v ] 
= \sum_{v \in T} \sum_{i=c+1}^n i \cdot \Pr[ L_j = i \mid W_j = v]\Pr[W_j = v] \\
&\leq \sum_{v \in T} \sum_{i=c+1}^n {i\over 2^{\Omega(i)}} \Pr[W_j = v]
\leq \sum_{v \in T} {1\over 2^{\Omega(c)}} \Pr[W_j=v] 
\le {1 \over 2^{\Omega(c)}}\,, \\
\end{align*}  
where the summation is taken over all nodes $v$ in the decision tree $T$.
The computation shows $\sum_{j=1}^n \E[L_j] \leq n/2^{\Omega(c)}$. 
By Markov's inequality,
we get that with probability at least $3/4$, we have $\sum_{j=1}^n L_j
\leq n/2^{\Omega(c)}$.
Thus, when the function score reaches $n/3$,
short edges must account for $n/3-n/2^{\Omega(c)}$ of the increase which is at least
$n/6$ for a large enough constant $c$.
Since any short edge has length at most $c$, there must be
at least $n/(6c)$ short edges. As discussed, this implies existence of 
$\Omega(n)$ active candidate sets that have size $O(1)$, meaning, their contribution
to the potential function is $\Omega(\log\log n)$ each. We have thus shown:

\begin{lemma}
\label{lem:highpotential}
Let $\vs$ be the random variable giving the leaf node of $T$ that the
deterministic query scheme ends up in on input $\Pi$. We have
$\phi(\vs)=\Omega(n\log\log n)$ with probability at least $3/4$.
\end{lemma}

Finally, we show how Lemma~\ref{lem:pot} and
Lemma~\ref{lem:highpotential} combine to give our lower bound. Essentially this boils down to showing that if the query scheme is too efficient, then the query asked at some node of $T$ increases the potential by $\omega(1)$ in expectation, contradicting Lemma~\ref{lem:pot}. To show this explicitly, define $\tsLB$ as the random variable giving the number of queries
asked on input $\Pi$. We have $\E[\tsLB]=t$, where $t$ was the expected number of queries needed for the deterministic query scheme. Also let
$\ell_1,\dots,\ell_{4t}$ be the random variables giving the first $4t$ nodes of
$T$ traversed on input $\Pi$, where $\ell_1 = r$ is the root node and
$\ell_i$ denotes the node traversed at the $i$th level of $T$. If
only $m<4t$ nodes are traversed, define $\ell_{i} = \ell_m$ for $i>m$; i.e., $\phi(\ell_i)=\phi(\ell_m)$.  From Lemma~\ref{lem:highpotential}, Markov's inequality and a union bound, we may now write
\begin{align*}
\E[\phi(\ell_{4t})] 
&= \E\left[\phi(\ell_1) + \sum_{i=1}^{4t-1} \phi(\ell_{i+1})-\phi(\ell_i)\right] 
= \E[\phi(r)] + \E\left[\sum_{i=1}^{4t-1} \phi(\ell_{i+1})-\phi(\ell_i)\right] \\
&= \sum_{i=1}^{4t-1} \E[\phi(\ell_{i+1})-\phi(\ell_i)]
= \Omega(n\log\log n).
\end{align*}
Hence there exists a value $i^*$, where $1 \leq i^* \leq 4t-1$, such that
\[
\E[\phi(\ell_{i^*+1}) - \phi(\ell_{i^*})] = \Omega(n\log\log n/t).
\]
But
\begin{eqnarray*}
\E[\phi(\ell_{i^*+1}) - \phi(\ell_{i^*})] =
\sum_{v \in T_{i^*}, v \textrm{ non-leaf}} \Pr[\Pi \in S_v] \E[\phi(w_{\is_v})-\phi(v) \mid \Pi \in S_v],
\end{eqnarray*}
where $T_{i^*}$ is the set of all nodes at depth $i^*$ in $T$,
$w_0,\dots,w_n$ are the children of $v$ and $\is_v$ is the random
variable giving the score of $F(x_v)$ on an input $\Pi \in S_v$ and $0$ otherwise. Since the events $\Pi \in S_v$ and $\Pi \in S_u$ are disjoint for $v \neq u$, we conclude that there must exist a node $v \in T_{i^*}$ for which
\[
\E[\phi(w_{\is_v})-\phi(v) \mid \Pi \in S_v] = \Omega(n \log \log n/t).
\]
Combined with Lemma~\ref{lem:pot} this shows that $n \log \log n/t =
O(1)$; i.e., $t = \Omega(n \log \log n)$. This concludes the proof of Theorem~\ref{thm:lower}.

%

\subsection*{Acknowledgments}
Most of this work was done while Benjamin Doerr was with the Max Planck Institute for Informatics (MPII) in Saarbr\"ucken, Germany, and Carola Doerr was with the MPII and Universit\'e Diderot, Paris, France.

Carola Doerr acknowledges support from a Feodor Lynen postdoctoral research fellowship of the Alexander von Humboldt Foundation and the Agence Nationale de la Recherche under the project ANR-09-JCJC-0067-01.

\newcommand{\etalchar}[1]{$^{#1}$}

\end{document}